\documentclass[12pt]{article}
\usepackage{amssymb}
\usepackage{amsthm}
\usepackage{amsmath}
\usepackage{booktabs}
\usepackage{graphicx}
\usepackage{natbib}
\usepackage{breqn}
\usepackage{dsfont}
\usepackage{multirow}
\usepackage{bbold}
\usepackage{geometry}
\usepackage{esint}
\usepackage{setspace}
\usepackage{caption}
\usepackage{xcolor}
\usepackage{hyperref}
\definecolor{linkcolour}{rgb}{0.2,0.2,0.6}
\hypersetup{colorlinks=true, breaklinks, urlcolor=linkcolour, citecolor=blue}

\newtheorem{theorem}{Theorem}[section]
\newtheorem{lemma}{Lemma}[section]
\newtheorem{proposition}{Proposition}[section]

\newtheorem{remark}{Remark}[section]
\theoremstyle{definition}
\newtheorem{definition}{Definition}[section]
\newtheorem{assumption}{Assumption}[section]

\title{Improving control over unobservables with network data}

\author{Vincent Starck\thanks{I am very grateful to Susanne M. Schennach for her helpful comments and continuous support throughout this project. I also thank Peter Hull, Toru Kitagawa, Soonwoo Kwon, Jonathan Roth, and Daniel Wilhelm for their valuable advice and numerous suggestions. This research has also benefited from the input of seminar participants at Brown, Michigan State, University of Pennsylvania, University of Alberta, University of Melbourne, University of Sydney, LMU Munich, and University of Amsterdam, as well as participants to the Microeconometrics class of 2022-2023 conference at Duke and the Encounters in Econometrics conference at Oxford.}\thanks{%
V.Starck@lmu.de}  \\ LMU Munich}
\date{\today}

\newcommand{\sumin}{\sum_{i=1}^n} 
\newcommand{\sumkn}{\sum_{k=1}^n}
\newcommand{\sumjn}{\sum_{j=1}^n}
\newcommand{\meanin}{\frac{1}{n} \sum_{i=1}^n}

\newcommand\equaldef{\mathrel{\overset{\makebox[0pt]{\mbox{\normalfont\tiny\sffamily def}}}{=}}}

\DeclareMathOperator{\CATE}{CATE}
\DeclareMathOperator{\ATE}{ATE}
\newcommand{\indep}{\perp \!\!\! \perp}

\begin{document}\maketitle

\begin{abstract}

This paper develops a method to conduct causal inference in the presence of unobserved confounders by leveraging networks with homophily, a frequently observed tendency to form edges with similar nodes. I introduce a concept of \textit{asymptotic homophily}, according to which individuals' selectivity scales with the size of the potential connection pool. The resulting network formation model accommodates common empirical features such as homophily, degree heterogeneity, sparsity, and clustering, and provides a framework to obtain consistent estimators of treatment effects that are robust to selection on unobservables. In an application, I recover an estimate of the effect of parental involvement on students' test scores that is greater than that of OLS, arguably due to the estimator's ability to account for unobserved ability. \medskip

\end{abstract}

\section{Introduction}

Estimating the effect of a treatment is a frequent goal in economics and social sciences. A common challenge is the presence of unobserved confounders that threaten the validity of the unconfoundedness assumption, which is typically necessary to perform inference with standard methods. Tools that strengthen control over variables that affect outcomes but are hard to measure, such as ability, culture, work ethic, tastes, \textit{etc.}, are thus particularly valuable. \medskip

Recently, networks and datasets with spatial structure have become increasingly available to researchers, providing new avenues for research. Homophily or assortative matching is a ubiquitous feature of empirical networks: nodes tend to associate with similar nodes \citep{lazarsfeld1954friendship, clark1992friendship, case1993budget, mcpherson2001birds, moody2001race, currarini2009economic, boucher2017my, dzemski2019empirical}. As \citet{zeleneev2020identification} note, homophily is likely to also operate through unobserved factors. An example is the tendency for people to form friendship ties based on ability \citep{clark1992friendship, burgess2011school, boutwell2017general}, a variable that is typically unavailable to the researcher. \medskip

Homophily then generates opportunities to create unobservable-adjusted comparison groups. For instance, unobserved ability is often a major concern in empirical studies that aim to assess the impact of a treatment on educational or wage outcomes. However, if people of similar ability are likely to be friends \citep{clark1992friendship, burgess2011school, boutwell2017general}, the omitted variable bias can be reduced by using friendship ties. In labor economics, coworker or reference networks often contain crucial information about unobserved productivity \citet{hensvik2016social} which can improve causal inference on job training or workplace interventions. \bigskip


The paper develops the idea that homophilic networks can be exploited to derive consistent estimators of treatment effects in the presence of unobserved confounders. To formalize this idea, I let the probability of link formation vary with the size of the network: people become pickier to limit the number of connections or improve their average quality as the network expands. This is consistent with the common view that the average degree should not increase proportionally to the size of the network and that most networks are sparse. As people are able to form increasingly better matches with a larger pool of potential neighbors, they become more selective because of decreasing benefits per additional match, preference for quality of matches, or limited resources to devote to additional connections. Improved match quality has been documented in a few instances, \textit{e.g.}, \citet{dauth2022matching} analyze worker sorting and find evidence of stronger assortative matching in larger cities. \medskip 

This modeling strategy accomplishes two things. First, the approach provides an asymptotic approximation that does not render the mechanism of network formation negligible in the limit: selectivity scales with the size of the connection pool, in the spirit of drifting sequences. As such, it proposes a network formation model that can accommodate common features of empirical networks such as homophily, sparsity, degree heterogeneity, or clustering. Second, it is sufficient to establish consistency and asymptotic normality for estimators that use $m$th-order connections or people with more than $c$ connections in common as comparison groups. \medskip 

As an extension, I discuss how homophily in dense networks can allow estimation under latent confounding by using comparison groups that consist of people whose observables increasingly differ but nevertheless connect. This hinges on the following intuition: if there is no observed rationale for two people being friends, the reason for their friendship is more likely to lie in the unobserved world. If two people are connected despite their observables indicating that such a link was unlikely, they are more likely to be close in terms of unobservables. By suitably manipulating a discrepancy in observables and letting it grow with sample size, consistent estimators can be constructed. \bigskip

I provide results that allow for the estimation of the Conditional Average Treatment Effect (CATE), which provides a way to describe the heterogeneity of the treatment effect for sampled individuals. Although the conditional average effect may be the end goal of the analysis (when a specific unit is targeted for treatment or policy), it is generally a prelude to aggregation to the Average Treatment Effect (ATE). \medskip

I define a general form of CATE estimator as a function of a group of counterfactual observations to be determined, then propose different choices to deal with different empirical issues. In all cases, estimators isolate increasingly better counterfactuals as to recover the CATE asymptotically. I show that the proposed estimators of the (C)ATE are asymptotically normal, enabling statistical inference. \medskip

Finally, I demonstrate the feasibility and effectiveness of the method through both simulations and an empirical application. In the application, I obtain an estimate of the effect of parental involvement on students' test scores that suggests a greater impact than OLS does, arguably due to the estimator's ability to account for unobserved ability and motivation.

\paragraph{Related literature}

The paper is at the intersection of the literature on networks \citep{jackson2010social, graham2015methods, de2017econometrics, newman2018networks}, in particular those featuring homophilic network formation \citep{boucher2015structural, graham2016homophily, graham2017econometric, demirer2019partial, gao2020nonparametric, mele2022structural}, and estimation of treatment effects \citep{imbens2004nonparametric, imbens2009recent, imbens2015causal}. \medskip

In a related paper, \citet{auerbach2022identification} considers a partially linear outcome regression where the nonlinear term depends on an unobserved variable. Using information from a network whose formation hinges on the unobserved variable, he is able to recover consistent estimates of regression coefficients under general assumptions. \citet{xu2025nonparametric} extends the control function approach to nonparametric outcome equations and derives minimax rates for estimation. See also \citet{goldsmith2013social, hsieh2016social, johnsson2021estimation}, who use related frameworks and provide a way to analyze peer effects. \medskip

Through the help of a pseudo-distance, \citet{zeleneev2020identification} devises a method to identify agents with similar values of latent fixed effects, which allows him to estimate parameters of interest while controlling for unobserved heterogeneity. \citet{demirer2019partial} provides partial identification results in linear models under homophilic behavior and proposes a comprehensive nomenclature for homophily. \medskip

This paper provides causal inference tools to handle latent confounding using information from a sparse network. Specifically, I consider a nonparametric potential outcome setup, but I impose structure on network formation, especially homophily in the unobservables. Leaving the outcome equation unrestricted is critical to deal with the common concerns of treatment effect heterogeneity and nonlinearities. In addition, the method circumvents the need to define and estimate equivalent classes and focuses on the common case of sparse networks, in contrast to previous papers. Finally, homophilic structures allow for the use of higher-order neighbors or friends in common through triangular inequality relationships, which leads to a class of intuitive estimators that are easy to implement. \medskip


\section{Improving control over unobservables using network data}

\subsection{Notation and baseline assumptions}

The sample is a cross-section of $n$ individuals. The treatment status of individual $i$, $T_i \in \{0, 1\}$, and the corresponding outcome, $Y_i = Y_i(T_i)$ with the potential outcome notation \citep{neyman1923application, rubin1974estimating}, are observed. As the notation for the outcome suggests, the Stable Unit Treatment Value Assumption (SUTVA) is maintained throughout; the setting is thought as a standard treatment effect problem with some confounders that cannot be measured, but some network information can be obtained to supplement the dataset. \medskip

Although homophilic restrictions still carry identifying power in the presence of some social interactions, adapting estimators and establishing their properties is to be done on a case by case basis. For instance, in the presence of peer effects, the estimators developed in the paper can be consistent for direct conditional treatment effects, \textit{i.e.}, the estimand has a causal interpretation but does not account for social multiplier effects.\footnote{See Appendix B for a simple illustration in the linear case.} 
In case of interference \citep{toulis2013estimation, sofrygin2016semi, aronow2017estimating, liu2019doubly, forastiere2020identification, sanchez2021spillovers, huber2021framework, leung2022causal}, the estimation strategy would have to be adjusted. Under the common assumption that only the fraction of treated friends matters, the comparison groups proposed in the paper could be further decomposed according to the fraction of treated connections. This could enable the estimation of treatment effects conditional on covariates and various levels of interference. 
However, accounting for more complex interference patterns, separating  direct and spillover effects, and obtaining proper inference results is likely to be a challenging task as such problems lie at the forefront of the interference literature \citep{leung2022causal}. \medskip


The covariates, $X = (X^o, X^u) \in \mathcal{X}^o \times \mathcal{X}^u \equaldef \mathcal{X} \subset \mathds{R}^d$, are divided into observed variables, $X^o$, and unobserved variables, $X^u$. There is a norm $\Vert \cdot \Vert$ on $\mathds{R}^d$ (with some abuse of notation, this will be used to represent the norm on $\mathcal{X}^o$ or $\mathcal{X}^u$), typically Euclidean. I focus on continuously distributed covariates $X$, though discrete variables can be accommodated--typically under weaker conditions since concerns such as asymptotic bias disappear. To avoid technical difficulties with vanishing denominators, it will be convenient to assume that covariates have a smooth density bounded from below. I make the following assumption throughout the analysis:

\begin{assumption}[Existence of bounded densities] \ \\
The joint distribution of the covariates admits a density $f$ with respect to Lebesgue measure. On the compact $\mathcal{X}$, the density is continuously differentiable and satisfies $f \geq \underline{f}$ for some positive $\underline{f}$.
\end{assumption}

The draws of $(Y_i, T_i, X_i)$ are i.i.d. and the realizations of a random variable are denoted by the corresponding lower-case letter. $B_r(x)$ denotes a ball of radius $r$ centered at $x$. $C$ represents a generic (positive) constant. \medskip 

A network is given through a (binary) weighting/link matrix $W$, of size ($n \times n$). The neighborhood $\mathcal{N}(i)$ refers to the links, friends, or connections of the node or individual $i$, \textit{i.e} $\mathcal{N}(i) \equaldef \{j \in \{1, \ldots, n\}\vert W_{ij} = 1\}$, $\mathcal{N}_t(i)$ denotes neighbors with a specific treatment status $t$, \textit{i.e} $\mathcal{N}_t(i) \equaldef \{j \in \{1, \ldots, n\}\vert W_{ij} = 1, T_j = t\}$. These definitions extend to higher-order neighbors, say of order $m$, which are denoted by $\mathcal{N}_t^m(i)$. Connections in common are given by $\mathcal{N}_t(i;j) \equaldef \mathcal{N}_t(i) \cap \mathcal{N}_t(j)$. \medskip 

The goal is to do inference about treatment effects. In particular, I develop inference methods for the Conditional Average Treatment Effect (CATE), $\CATE(x_i) \equaldef \mathds{E}[Y_i(1)-Y_i(0)\vert X_i=x_i]$, and then for the Average Treatment Effect (ATE), $\ATE \equaldef \mathds{E}[Y_i(1)-Y_i(0)]$. Although I focus on average treatment effects, the insights can be exploited to obtain, \textit{e.g.,} quantiles of treatment effects or the average effect on the treated. The usual statement about omitting ‘almost surely’ qualifiers, in particular pertaining to conditional expectations, applies. \medskip 

The following core assumptions are maintained throughout the paper: 

\begin{assumption}[Causal Inference] \ \\
a) Unconfoundedness: $(Y_i(1), Y_i(0)) \indep T_i \vert X_i$ \\
b) Overlap: $0<C<\mathds{P}[T_i = 1\vert X_i] < 1-C < 1$
\end{assumption}

These two assumptions are ubiquitous in the literature on treatment effects, although this version of unconfoundedness conditions on $X$ instead of $X^o$. It is only assumed that treatment is independent of potential outcomes when conditioned on individual characteristics, including unobserved ones. Since covariates that may influence selection into treatment, such as ability, work ethic, or personal preferences, are typically unobserved, this is often a valuable relaxation: selection on some unobservables is allowed.

\subsection{Network formation}

Let $i, j$ be two individuals and $i \neq j$. I focus on link-formation models of the type 
\begin{equation}\label{network_formation}
    W_{ij} = 1 \iff \eta_{ij} \leq w_n(h(X_i^o; X_j^o) + \Vert X_i^u - X_j^u\Vert)
\end{equation} 
where $w_n: \mathds{R}^+ \rightarrow [0;1]$ is a decreasing function that satisfies $\lim_{x \rightarrow \infty} w_n(x) = 0$. Typically, $w_n$ would decrease with $n$ to accommodate network sparsity, \textit{e.g.}, $w_n(x) = \max\{1-s_n x, 0\}$ or $e^{-  \frac{1}{2} (s_n x)^2}$. The function $h$ is arbitrary and could depend on $n$ as well, and $\eta_{ij}=\eta_{ji}$ are independent uniform\footnote{Since any inverse cumulative distribution function can be applied on both sides to generate any distribution, the uniform assumption is just a normalization.} shocks, drawn independently of $(X_i, X_j, T_i, T_j, Y_i, Y_j)$. The dimensionality of the unobserved variables is arbitrary and matters only for some rates of convergence. \medskip

Dyadic network formation processes are common in the literature\textit{, e.g.}, \citet{graham2017econometric, gao2020nonparametric, zeleneev2020identification, auerbach2022identification, johnsson2021estimation}. Compared to more general specifications (for example, \citet{auerbach2022identification} posits that links are formed whenever $\eta_{ij} \leq w(X_i, X_j)$ and only imposes a weak continuity assumption on $w$), the model (\ref{network_formation}) adds some separability and homophily in the unobservables. \medskip

\begin{remark}
    
There may be another type of unobservable that reflects factors such as popularity or expansiveness. 
Such factors can be accommodated by replacing the right-hand side of (1) with $\tilde{w}_n\left(h(X_i^o; X_j^o) + \Vert X_i^u - X_j^u\Vert, A_i, A_j\right)$, where $A_i$ and $A_j$ are unobserved heterogeneities that represent expansiveness.\footnote{Since links are undirected, $A_i$ and $A_j$ should enter $\tilde{w}$ symmetrically.} \medskip

The model can then be interpreted as extending \citet{graham2017econometric}'s (see also \citet{dzemski2019empirical} in the case of directed links): compared to models where links are formed whenever $X_{ij}'\theta + A_i + A_j$ exceeds (logistic) errors, the model relaxes slightly the functional form and distribution of errors, and introduces unobserved components of $X$ that satisfy homophilic restrictions. The two frameworks coincide if we parametrize $h(X_i, X_j) + \Vert X_i^u - X_j^u\Vert \equiv X_{ij}' \theta$ where $X_{ij}$ is obtained from a transformation of $(X_i, X_j)$ such that the unobserved components are featured homophilically, and we let $\tilde{w}(z, a, b) \equiv \mbox{Logit}(z+a+b)$. \medskip

On the other hand, heterogeneities can be introduced in new ways. A variant that may be of interest specifies the argument of the link function as $\left(\frac{h(X_i^o; X_j^o) + \Vert X_i^u - X_j^u\Vert}{(A_i A_j)^{1/d}}\right)$. Here, expansiveness is featured in a multiplicative way, which may be easier to interpret as a selectivity scale. In the limit of the asymptotic homophily model to be developed in the next Section, it would correspond to a scaling factor for probabilities. A person with characteristics $(x_i, 2 a_i)$ is twice as expansive and has twice the probability of forming a link than a person with characteristics $(x_i, a_i)$ as $n \rightarrow \infty$. \medskip

The analysis can be applied to such functions $\tilde{w}$ under the assumption that expansiveness does not affect the outcome of interest, conditional on $X$. The results presented in the following sections apply directly if those adjustments are made: 
\begin{enumerate}
    \item In the asymptotic framework of Section 2.3, the dependence on $n$ is through the first argument of $\tilde{w}$; in Graham's model, this is interpreted as letting $\theta$ drift at some rate $s_n$: links are formed whenever $X_{ij}' (\theta s_n)$ is lower than $U_{ij} + A_i + A_j$, where $U$ is logistic and the unobserved heterogeneities may be correlated with $X$;
    \item In the dense framework extension in Appendix B, the function $\tilde{w}(z, a, b)$ can be lower bounded by a function $\underline{\tilde{w}}(z)$ that shares the properties required of $w$ in that section; in Graham's model, this is interpreted as imposing bounds on $A_i$.
\end{enumerate}


\end{remark}

Another feature of this network formation is the explicit dependence of $w$ on network size, allowing for sparse networks. This is often the empirically relevant setup since the mean degree of a node is rarely expected to scale with the size of the network \citep{jackson2010social, newman2018networks}. \medskip

The model can be given the usual interpretation of ‘link creation under a mutual positive utility of forming a link’ ($w-\eta$ then reflecting utility; see, \textit{e.g.}, \citet{jackson2010social}), where people derive more utility from interacting with similar individuals, or rationalizes the idea that people with similar characteristics are more likely to meet and thus to form a connection. Nevertheless, since the network is primarily seen as information to draw from, this rationale may not be necessary. For instance, if individuals end up developing similar characteristics after randomly forming connections, a researcher who observes the network after covariates have evolved could use the present framework. In other words, (\ref{network_formation}) need not be the structural equation for network formation, but should approximate the relationship between the links and the covariates relevant to selection at the time of observation. \medskip

The function $h$ may also exhibit homophily, in which case $h(X_i^o; X_j^o) = \Vert X_i^o - X_j^o\Vert$, but other forms may be more arise in specific applications. A contrasting example would be a search for skill complementarity, \textit{e.g.}, for work relationships, inducing non-homophilic selection in a covariate. Such complementarity matching may create issues if the variable on which it operates belongs to the set of controls because, at least in its extreme form, it creates a tendency to match only with increasingly distant individuals, which makes it difficult to handle without functional form restrictions. However, as long as a friendship link does not increasingly rule out closeness in covariates, smoothing on observables allows us to control for the observed components regardless of how they are featured in network formation. \medskip

The following assumption is thus imposed to ensure the possibility of smoothing on the desired observed components.\footnote{An alternative would be to combine the approach with functional form restrictions, though that may be less appealing in practice.} To simplify notation, it is stated as if all variables entering $h$ are intended to be in the control set, but variables that do not enter the outcome equation (even if unobserved) could be featured arbitrarily in $h$. 

\begin{assumption}\label{h_assumption}
    $h$ is continuously differentiable in both of its arguments and, for all $x_i$, there exists $\varepsilon>0, \delta>0$ such that $f_{X_j^o \vert W_{ij} = 1, X_i = x_i}(y) \geq \varepsilon$ for all large $n$ and $y \in B_\delta(x_i^o)$. 
\end{assumption}

The assumption requires that creating a link does not completely rule out similar covariates. A way to ensure this is to assume people do not strongly penalize their own type: $h(x, x) = h_n(x, x) \rightarrow 0$ as $n$ grows, which includes cases such as $h(x, \tilde{x})=0 \ \forall \tilde{x} \in \mathcal{T}_x \subseteq \mathcal{X}^o$, where $\mathcal{T}_x$ is a list of 'preferred types' which includes $x$ and $h$ is otherwise arbitrary, or when $h(x, \cdot)$ is suitably scaled by a function of $n$ to achieve some preference distribution over characteristics whose support includes $x$. Of course, observables that affect the outcome but do not affect network formation are easily accommodated (they may be included as a special case with an additive component of $h$ which is identically $0$). What is however ruled out is anti-homophilic matching. 

\subsection{Asymptotic homophily}\label{section_sh}

\subsubsection{The asymptotic homophily framework}

In what follows, it is assumed that the vector $X$ from the unconfoundedness condition is part of the network formation in a homophilic way to facilitate the exposition. In practice, the binding restriction is that $X^u$ satisfies this assumption since it is possible to smooth over observables (\textit{e.g.}, Remark \ref{kernel_product_remark} below) under Assumption \ref{h_assumption}. \medskip

I explore the case of \textit{asymptotic homophily}, \textit{i.e.}, homophilic behavior is the core mechanism of network formation and individuals' selectivity is tied to the size of their potential matching pool. The improvement in matching quality provides one rationale for sparse networks.\footnote{An extension, discussed in Subsection \ref{section_hiu} in the Appendix, allows for a constant match quality by possibly letting the link formation be independent of sample size; this setup accommodates denser networks.} \medskip

This section provides an asymptotic theory when homophilic behavior is pronounced relative to network size and formalizes the intuition that connections among individuals can be used to form comparison groups. It shows the identifying power of homophilic restrictions and provides a network formation model compatible with many empirically relevant features. \medskip

Suppose that associations are captured by homophily, the probability of a connection is decreasing in $\Vert X_i - X_j\Vert$, and the network is sparse: the average degree of a node is constant or increases slowly. As people are able to form increasingly better matches with a larger pool of potential neighbors, they become more selective because of decreasing benefits per additional match, preference for quality of matches, or limited resources to devote to additional connections. \medskip

To reflect this behavior, the sequence of functions $w_n$ must satisfy two conditions. First, the sequence must be decreasing in order to decrease the probability of forming connections as $n$ rises. Homophily further suggests that people penalize dissimilar individuals increasingly more harshly so that the average match quality (in terms of homophilic preferences) increases. \medskip

The improvement in match quality is consistent with empirical evidence in some applications, such as stronger assortative matching in larger cities \citep{dauth2022matching}. Another interpretation is that asymptotic homophily formalizes the notion that homophilic behavior is pronounced relative to sample size in the spirit of a drifting sequence\footnote{Similarly to \citet{bekker1994alternative}, who analyzes the behavior of IV estimators with many instruments, or \citet{borusyak2022quasi}, who analyze shift-share instruments with a growing number of shocks. ``The sequence is designed to make the asymptotic distribution fit the finite sample distribution better. It is completely irrelevant whether or not further sampling will lead to samples conforming to this sequence'' \citep{bekker1994alternative}.}; the degree to which individuals are selective is, in a sense, preserved as we proceed to an asymptotic approximation. \medskip

Functions of the form $w_n(x) \geq g(s_n x)$ are consistent with such behavior---homophily becomes more prevalent as $n$ rises---irrespective of the exact form of $w_n$ (or $g$). I adopt the following definitions: \medskip 

\begin{definition}[Asymptotic homophily]
a) Network formation is asymptotically homophilic if $W_{ij} = \mathds{1}_{\eta_{ij} \leq w_n(\Vert X_i - X_j\Vert)}$, $w_n(x) \geq g(s_n x)$, where $g: [0; \infty[ \rightarrow [0;1]$ is decreasing and $\lim_{n \rightarrow \infty} s_n = \infty$. \\
b) Network formation is regularly asymptotically homophilic if $W_{ij} = 1$ whenever $w_n(\Vert X_i - X_j\Vert) \geq \eta_{ij}$, $w_n(x) =g(s_n x)$, where $g: [0; \infty[ \rightarrow [0;1]$ is a decreasing function such that $0 < \int_{\mathds{R}^d} g(\Vert y\Vert) \ dy < \infty$, and $\lambda \equaldef \lim_{n \rightarrow \infty} \lambda_n \in \ ]0, \infty]$ where $\lambda_n \equaldef n {s_n}^{-d}$.
\end{definition}

Part a) formalizes the idea of asymptotic homophily; part b) provides regularity conditions for asymptotic results. $\lambda \in ]0, \infty]$ accommodates both the ``exactly sparse'' case when average degrees stay constant ($\lambda \in \mathds{R}^+$), and the slowly increasing average degree regimes ($\lambda =\infty$); the rate condition basically rules out regimes in which average degrees vanish asymptotically. 

When the scaling factor $s_n$ induces bounded degrees, the model is similar to \citet{leung2019weak} who studied such models in the presence of strategic interactions. It generally corresponds to geometric random graph structures \citep{penrose2003random}. 
The resulting formation mechanism is consistent with the empirical regularities of social networks \citep{jackson2010social}: sparsity\footnote{It is natural to let the average degree of a node be constant or grow only slowly for most applications. For instance, the average number of friends is typically viewed as constant or slowly increasing as the network expands, requiring the probability of forming a link to decrease with the size of the network. Letting the degree increase, albeit slowly, allows one to take advantage of asymptotic approximations while constant degree is often of interest and considered as the exactly sparse case.}, transitivity/clustering\footnote{Intuitively, clustering occurs because groups of similar individuals tend to form connections. Moreover, as shown in the appendix, the clustering coefficient does not vanish asymptotically, in contrast to Poisson random graphs \citep{erdHos1960evolution} or configuration models \citep{bender1990asymptotic}.}, degree heterogeneity\footnote{Because of the influence of covariates, the expected degree varies across individuals. The underlying density affects the degree distribution because people with common characteristics have an easier time forming connections. Another source of degree heterogeneity can be accommodated for by adding expansiveness factors as previously discussed--for instance by refining the sequence $s_n$ with pair-specific rates using the multiplicative heterogeneity model, although this is not pursued here for simplicity.}, and homophily. \medskip


\subsubsection{Comparison groups}

Under an asymptotically homophilic network formation process, it is possible to derive estimators whose bias is asymptotically negligible, even if some confounders are unobserved. Given a comparison group $\mathcal{C}_i$ for individual $i$, I define a CATE\footnote{Note that $x_i$ is not fully observed. The CATE of a given individual is identified, but not the underlying function of $x$. For this reason, the CATE is mainly interesting on its own when one cares about the treatment effect of a specific unit.} estimator 
\begin{equation}
    \widehat{\CATE}(x_i; \mathcal{C}_i) \equaldef \frac{1}{\vert \mathcal{C}_{i1}\vert} \sum_{j \in \mathcal{C}_{i1}} Y_j - \frac{1}{\vert \mathcal{C}_{i0} \vert} \sum_{j \in \mathcal{C}_{i0}} Y_j
\end{equation}
where $\mathcal{C}_{it} \equaldef \mathcal{C}_{i} \cap \{j\vert T_j = t\}$. Note that the comparison group will vary with the sample size, although the dependence is left implicit in the notation. \medskip

\begin{remark}\label{kernel_product_remark}
If some observed covariate ($x_k$) that affects the outcome does not influence network formation, it can be controlled for nonparametrically by multiplying each weight $\mathds{1}(j \in \mathcal{C}_{it})$ with kernel weights $K_b(\Vert X_{ik}^o - X_{jk}^o\Vert)$ throughout, or performing (local) regression adjustments within comparison groups as described in the simulation study. Under Assumption \ref{h_assumption}, this strategy is also valid for covariates that influence network formation in a non-necessarily homophilic way. The following discussion can then be interpreted as the complementary results that ensure that the key unobserved components are controlled for using network information.
\end{remark}

The first idea is to rely on friends to construct a comparison group, \textit{i.e.}, set $\mathcal{C}_i=\mathcal{N}(i)$. This yields a CATE estimator based on the difference between treated friends and non-treated friends. Thanks to triangular inequality relationships and the nature of homophily, however, one can extract additional information from the friendship network. For instance, one can consider friends of friends or higher-order friendships: 

\begin{align}\label{strongCATEm}
\begin{split}
    \widehat{\CATE}(x_i; \cup_{m=1}^M \mathcal{N}^m(i)) & = \frac{1}{\vert \cup_{m=1}^M \mathcal{N}_1^m(i)\vert} \sum_{j \in \cup_{m=1}^M \mathcal{N}_1^m(i)} Y_j \\ 
    & - \frac{1}{\vert \cup_{m=1}^M \mathcal{N}_0^m(i) \vert} \sum_{j \in \cup_{m=1}^M \mathcal{N}_0^m(i)} Y_j,
\end{split}
\end{align}
for some upper order of friendship $M$. For $M=1$, this is the simple estimator that compares treated friends and non-treated friends. Although the estimator averages more observations as $M$ increases, the comparison group increasingly selects observations whose characteristics differ from those of $i$. \medskip

An alternative estimator relies on having at least $c$ friends in common:
\begin{align}\label{strongCATEtau}
\begin{split}
    \widehat{\CATE}(x_i; \{j\vert \ \vert \mathcal{N}(i;j)\vert \geq c\}) & = \frac{1}{\vert  \{j\vert \ \vert \mathcal{N}_1(i;j)\vert \geq c\}\vert} \sum_{j \in \{j\vert \ \vert \mathcal{N}_1(i;j)\vert \geq c\}} Y_j \\ 
    & - \frac{1}{\vert \{j\vert \ \vert \mathcal{N}_0(i;j)\vert \geq c\} \vert} \sum_{j \in \{j\vert \ \vert \mathcal{N}_0(i;j)\vert \geq c\}} Y_j.
\end{split}
\end{align}

Although both estimators can allow for consistent estimation of treatment effects, the composition of the underlying comparison groups can differ significantly. Therefore, it can be a useful robustness check to compare their treatment effect estimates, for instance if one is concerned about peer effects or related issues that would likely affect these estimators in different ways. \medskip

As an illustration, consider Figure 1 where the comparison group for individual $i$ with unobserved characteristics $x_i \in \mathds{R}^2$ is shown in red and the remaining observations in black. The leftmost picture is the unknown\footnote{Except in the extreme network formation process in which individuals select friends deterministically conditional on covariates: $w(x) = \mathds{1}_{[0, C]}(x)$. Then, the first two pictures become identical.} group that consists of all observations below a certain distance. Next, on the right, friends are used as the comparison group, providing a noisy version of (i): selected observations tend to fall close to $x_i$, but some close observations are ignored while observations farther away may be selected nonetheless. \medskip

In the third picture, one looks at friends of friends. This allows us to make use of more observations, which will reduce the variance of the estimator, but there is also a tendency to grab more observations outside the sphere. Finally, the last picture selects individuals who have at least two friends in common with $i$. This typically reduces the bias compared to using the previous groups, but selects fewer observations. \medskip

\begin{figure}[ht!]
\centerline{\includegraphics[width=6in]{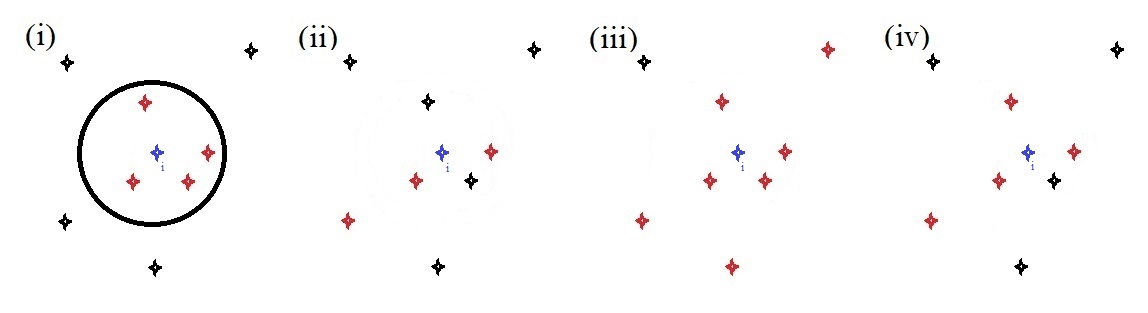}}
\caption{Comparison among possible comparison groups for individual $i$. Observations are represented by stars (blue = $i$; red = included in $\mathcal{C}_i$; black = not included). From left to right, 
(i) (Infeasible) individuals within a given distance of $x_i$, (ii) individuals who are friends with $i$, (iii) individuals who are friend with $i$ or friend of a friend of $i$, (iv) individuals who have two friends in common with $i$.}
\end{figure}

\subsubsection{CATE inference}

Under asymptotic homophily, standard asymptotics are achievable with estimators such as (\ref{strongCATEm}) or (\ref{strongCATEtau}), whose bias asymptotically disappears. As for the variance, its collapse is based on averaging over increasingly many observations. In what follows, I focus on the common case where the network contains a giant component, which comprises a large fraction of the nodes. To ensure this, I assume that $\lambda_n \psi_i > (m_c \vee 1) \ \forall i$ eventually, where $\psi_i \equaldef f(x_i) \int g(\lVert y \rVert) \ \mbox{d}y$ and $m_c$ is the critical mean degree to ensure the existence of a giant component \citep{penrose2022giant}. The condition is automatically satisfied when $\lambda = \infty$. 

The asymptotic behavior of the size of a comparison group depends on the speed of convergence of $w_n$, of which the following lemma provides a formal analysis. \medskip

\begin{lemma}\label{lemma}
The number of connections for $i$ exceeds any real number with probability approaching one if $\lim_{n \rightarrow \infty} n \int_{\mathds{R}^d} w_n(\Vert x_j - x_i\Vert) f(x_j) \ \mbox{d}x_j = \infty$. \\
Under regular asymptotic homophily, this holds when $\lambda = \infty$; the number of connections remains (a.s.) bounded when $\lambda$ is finite.

Moreover, \\
a) The probability of forming a connection of order up to $M$ is $O\left(\psi_i {s_n}^{-d} (\psi_i \lambda_n)^{M-1}\right)$ with $M = o(\ln(n))$. \\
b) If $n {s_n}^{-2d} \rightarrow 0$, the probability of $i$ and $j$ having at least $c$ friends in common is $O\left((n {s_n}^{-2d})^c\right)$. 
\end{lemma}

The lemma, proven in the appendix, provides conditions under which the sizes of potential comparison groups grow to infinity. It also specifies the rates at which the probabilities of forming a connection up to the $M$-th or having more than $c$ friends in common decrease under asymptotic homophily. With overlap, the lemma ensures the number of treated and untreated connections both grow to infinity when $\lambda = \infty$. In the exact sparsity case, \textit{i.e.} when $\lambda$ is finite, the number of connections remains naturally bounded. Although growing the size of the comparison group is still feasible, \textit{e.g.} by letting $M$ grow with the sample size, this case will require special care because identification may also be affected. \medskip

The condition in the lemma states that $w_n$ must not vanish too quickly to ensure that connections are still being formed. The formal criterion analyzes the integral $\int_{\mathds{R}^d}w_n(\Vert x\Vert) f(x+x_i) \ dx$, suggesting that link functions that do not depend on network size or vanish uniformly slowly enough such as $w_n(x) = {s_n}^{-1} g(x)$ induce non-trivial or unbounded friend counts. Basically, individuals must not become too selective too quickly to ensure that they keep forming connections. This is natural in many networks (\textit{e.g.}, friendship networks) as the expected degree is often viewed as, at best, only slowly increasing. \medskip

Asymptotic homophily also implies an improvement in matching quality that is absent when $w_n$ is constant or grows uniformly. Specifically, a sequence such as $w_n(x) = {s_n}^{-1} g(x)$ would stabilize the ``posterior'' distribution $f_{X_j \vert j \in \mathcal{N}(i)}$; it does not imply that people improve their average match in larger networks. As a result, \textit{regular asymptotic homophily} will be key in securing consistency properties. An important part in establishing these is the analysis of the bias 
\begin{align*}
    \begin{split}
        \mathds{B}_i & \equaldef \mathds{E}[Y_j(1) \vert j \in \mathcal{C}_i, T_j=1] - \mathds{E}[Y_j(1)\vert X_j = x_i] \\
        & - (\mathds{E}[Y_j(0) \vert j \in \mathcal{C}_i, T_j=0] - \mathds{E}[Y_j(0)\vert X_j = x_i])
    \end{split}
\end{align*}
which will be shown to disappear under various conditions. Specifically, the following conditions on $w_n$ will ensure that the bias vanishes. \medskip

\begin{assumption}[Hölder continuity of CATE and convergence of link function] \label{Hölder}
a) $\CATE(x)$ is Hölder continuous with exponent $\alpha$ on a neighborhood of $x_i$, \textit{i.e.} for any $x, y$ in the neighborhood $\Vert\CATE(y) - \CATE(x)\Vert \leq C \lVert y - x \rVert^\alpha$ for some $\alpha > 0$. \\
b) For some $\varepsilon_n \downarrow 0$, either $\mathcal{C}_i=\cup_{m=1}^M \mathcal{N}^m(i)$ and $\sum_{m=1}^M (w_n(\frac{\varepsilon_n}{m}))^m = o({(\psi_i \lambda_n})^M n^{-1})$ and $M=M_n \uparrow \infty$ (but slower than $\ln(n)$) if $\lambda$ is finite, or $\mathcal{C}_i=\{j\vert \ \vert \mathcal{N}(i;j)\vert \geq c\}$ and $(w_n(\frac{\varepsilon_n}{2}))^c = o({s_n}^{-d})$ and $\lambda = \infty$.
\end{assumption}

Hölder continuity is a standard assumption that imposes a mild degree of smoothness in the CATE. Part b) of the assumption restricts how $w_n$ decays from the origin; it requires a sufficiently fast convergence away from the origin. As discussed previously, ensuring consistency requires adjusting the comparison groups to ensure that their size grows in the exact sparsity case. A way to achieve this is to let $M=M_n$ grow at a suitable rate. \medskip 

Although consistency can be achieved under very weak conditions, the rate of convergence may be slow and the conditions on $w_n$ are hard to interpret. Assuming that the underlying functions of $X$ are sufficiently smooth implies a clean bias rate of $O({s_n/M}^{-2})$ under regular asymptotic homophily. Although this can be established without requiring more than second-order derivatives, I formulate a more general smoothness assumption that will prove useful when discussing asymptotic normality:

\begin{assumption}[Smoothness] 
The density $f$ is $L$ times continuously differentiable with uniformly bounded $L$-th derivatives and so are the propensity score $p(x) \equaldef \mathds{P}[T_i=1\vert X_i=x]$ and the conditional expectations $\mathds{E}[Y_i(t)\vert X_i=x]$ for $t=0, 1$.
\end{assumption}
Now, the consistency theorem reads
\begin{theorem}[Consistency]\label{CATE_consistency}

a) Suppose $\mathds{E}[Y_j(t)^2] < \infty$ for $t = 0, 1$. \\
Then, $\widehat{\CATE}(x_i; \mathcal{C}_i)$ is consistent for $\CATE(x_i)$ under Assumption \ref{Hölder} whenever $\mathds{E}[\lvert \mathcal{C}_i \rvert]$ grows. 
Moreover, the bias satisfies $\mathds{B}_i = O(\varepsilon_n^\alpha + R)$ with $R= {\lambda_n}^{-M} n \sum_{m=1}^M w_n(\frac{\varepsilon_n}{m})^m$ and $R={s_n}^d w_n(\varepsilon_n/2)^c$, respectively, for $\varepsilon_n \downarrow 0$ as in Assumption \ref{Hölder}. \\
b) If the network formation is regularly asymptotically homophilic, and Smoothness holds with $L >1$, then the estimators are consistent with bias $\mathds{B}_i = O({s_n}^{-2})$. If higher-order friendships are used and $M=M_n$, then the bias is of the order $O\left(M^2 {s_n}^{-2} \right)$.
\end{theorem}

The theorem is proven in the appendix. The main difficulty in part a) is to derive an expression for the bias that can subsequently be bounded via homophilic assumptions and triangular inequalities. In the second part, the existence of derivatives allows the use of Taylor expansions and the derivation shares similarities with nonparametric kernel analysis, though the noisy matching through $w_n$ makes the problem non-standard. \medskip 

The condition that $\mathds{E}[\lvert \mathcal{C}_i \rvert]=n \mathds{P}[j \in \mathcal{C}_i]$ grows is an important caveat. Consistency is not achievable for every individual in the exactly sparse regime because, asymptotically, there is a nontrivial probability that an individual does not form any connections (or that they only have friends with no other connections, etc.), which leads to a non-vanishing variance even if $M$ grows. This further affects the estimation of the ATE as discussed later, although a weighted average of the treatment effects is still consistently estimable. \medskip


I complete the analysis with an asymptotic normality result: CATE estimators are asymptotically normal at $x_i$. Formally, 

\begin{theorem}[Asymptotic Normality]\label{CATE_asymptotic_normality}
Suppose that the assumptions of the Consistency theorem hold and that potential outcomes have $2+\delta$ moments for some $\delta > 0$ (conditional on $X=x_i$). Then, for a bias $\mathds{B}_i$ at location $x_i$, the (conditional) asymptotic distribution reads
\begin{equation*}
    \sqrt{n \mathds{P}[j \in \mathcal{C}_i]} (\widehat{\CATE}(x_i; \mathcal{C}_i)-\CATE(x_i)-\mathds{B}_i) \overset{d}{\rightarrow} \mathcal{N}\left(0; V\right)
\end{equation*}
where $V = \frac{\mathds{V}[Y_j(1)|X_j=x_i]}{\mathds{P}[T_j=1\vert X_j = x_i]}+\frac{\mathds{V}[Y_j(0)|X_j=x_i]}{\mathds{P}[T_j=0\vert X_j = x_i]}$. \\
Moreover, $\mathds{B}_i$ is asymptotically negligible if $\CATE(x)$ is Hölder continuous with exponent $\alpha$ on a neighborhood of $x_i$ and one of the following holds: \\
(i) $\mathcal{C}_i=\cup_{m=1}^M \mathcal{N}^m(i)$ with $\sum_{m=1}^M w_n(\frac{n^{-\gamma}}{m})^m = o((\psi_i \lambda_n)^M n^{-1})$ for $\gamma > \frac{1}{2 \alpha}$ and $M$ grows if $\lambda < \infty$, or \\
(ii) $\mathcal{C}_i=\{j\vert \ \vert \mathcal{N}(i;j)\vert \geq c\}$ and ${s_n}^d w_n(n^{-\gamma})^c = o({\lambda_n}^{-1/2})$, $\lambda=\infty$, and $\gamma > \frac{1}{2 \alpha}$ or \\
(iii) Network formation is regularly asymptotically homophilic, Smoothness holds with $L>1$, and $\sqrt{\mathds{E}[\vert \mathcal{C}_{i}\vert]}{(M/s_n)^2} \rightarrow 0$. 
\end{theorem}

When the CATE itself is of particular interest, the theorem provides a way to perform standard inference. Perhaps more importantly, consistent estimation of treatment effects at given $x_i$'s suggests that inference about average effects is possible. This is the next result: the average treatment effect can be estimated under unobservable-robust unconfoundedness. 


\subsubsection{ATE inference}

Given the last two theorems, one obtains an estimator $\widehat{\ATE}$ by averaging over a CATE estimator at all $x_i$. The resulting ATE estimator is then consistent and asymptotically normal under regularity conditions. \medskip

Specifically, consider $\meanin \widehat{\CATE}(x_i; \mathcal{C}_i)$ and collect the terms involving $Y_i$ for each $i$ to obtain $\meanin \left(T_i \sum_{j \in \mathcal{C}_{i}} \frac{1}{\vert \mathcal{C}_{j1}\vert} - (1-T_i) \sum_{j \in \mathcal{C}_{i}} \frac{1}{\vert \mathcal{C}_{j0}\vert}\right) Y_i$.\footnote{Although the weights are well defined with probability approaching one, it may be useful to regularize them in finite sample by adding a small vanishing offset to the denominators.} The estimator can be further adjusted in the spirit of AIPW \citep{robins1994estimation} or imputation with regression adjustments \citep{rubin1973use, abadie2011bias, lin2025regression} to obtain a doubly robust version with better asymptotic guarantees. Specifically, consider

\begin{equation}\label{ATE_hat}
    \widehat{\ATE} \equaldef \meanin \hat{Y}_i(1) - \hat{Y}_i(0)
\end{equation}
where imputed potential outcomes are constructed as $\hat{Y}_i(t) = \sumjn \mathds{1}(T_j=t) \omega_{ijt} (Y_j - \hat{\mu}_t(X_j) + \hat{\mu}_t(X_i))$ with $\omega_{ij t} \equaldef \mathds{1}(j \in \mathcal{C}_{i}) \frac{1}{\lvert \mathcal{C}_{i t}\rvert}$ for $t\in \{0, 1\}$ and $\hat{\mu}_t$ is constructed as in the CATE estimators.\footnote{Alternatively, one can directly use $i$'s outcome to impute the corresponding potential outcome: $\hat{Y}_i(t) = \mathds{1}(T_i = t) Y_i + \mathds{1}(T_i=1-t) \sumjn \mathds{1}(T_j=t) \omega_{ijt} (Y_j - \hat{\mu}_t(X_j) + \hat{\mu}_t(X_i))$ as in \citet{lin2025regression}. The two versions share the same asymptotic properties. In finite sample, double imputation may suffer from missing counterparts more frequently, but could occasionally attenuate bias by treating both treatment arms symmetrically regardless of the treatment status $T_i$.}\medskip 


There are a few observations for which the comparison group is not defined since some individuals (or their higher-order connections) may not form any link. Those observations can be dropped without consequence for the asymptotics when $\lambda = \infty$, but their removal affects the estimand when $\lambda$ is finite. The following theorem, which describes the asymptotic distribution of the estimator, is thus formulated as an asymptotic normality for estimating a Weighted ATE (WATE). Let $\widehat{\mbox{WATE}}$ denote the estimator (\ref{ATE_hat}) trimmed of the observations that do not belong to the giant component $\mathcal{G}$ of the network, \textit{i.e.}, observations $i=1, \ldots, n$ are now interpreted as originating from the giant component. We have the following result:
\begin{theorem}\label{ATE_CAN}
Suppose that $\mathds{E}[Y_j(t)^2] < \infty$, that condition (iii) of Theorem 2.2 holds, that Smoothness holds for $L > 2$, that the function $g$ is continuously differentiable $L$ times with uniformly bounded $L$-th derivatives, and $\int g(\lVert y \rVert) \lVert y \rVert^L \mbox{d}y < \infty$. 
Then, if $n^{1/2} \inf_i (\psi_i \lambda_n)^{-M} \ln(n) (s_n/M)^{-\ell} \rightarrow 0$ for $\ell = 2, \ldots, L-1$ and $\sqrt{n} {(s_n/M)}^{-L} \rightarrow 0$,
\begin{equation*}
    \sqrt{n} (\widehat{\mbox{WATE}} - \mbox{WATE}) \overset{d}{\rightarrow} N\left(0; \mathds{E}\left[\frac{\mathds{V}[Y_i(1)\vert X_i]}{p(X_i)} + \frac{\mathds{V}[Y_i(0)\vert X_i]}{1-p(X_i)}\middle\vert \mathcal{G}\right] + \mathds{V}[\CATE(X_i)\vert \mathcal{G}]\right),
\end{equation*}
where $\mbox{WATE} \equaldef \mathds{E}[\mathds{E}[Y_i(1)-Y_i(0)\vert X_i] \vert \mathcal{G}]$.
A consistent estimator of the asymptotic variance is given by
\begin{align}
\begin{split}
    \meanin \Biggl[(2 T_i-1) \left(\sumjn T_j \omega_{ji1} + (1-T_j) \omega_{ji0} \right) (Y_i - T_i \hat{\mu}_1(X_i) - (1-T_i) \hat{\mu}_0(X_i)) \\
    + \hat{\mu}_1(X_i) - \hat{\mu}_0(X_i) - \widehat{\mbox{WATE}}\Biggr]^2 .
\end{split}
\end{align}
\end{theorem}

The theorem is proven in the appendix. It can be seen that the asymptotic variance has the form of the semiparametric efficiency bound for estimating the ATE (\citet{hahn1998role, hirano2003efficient}), but with averages taken on the giant component subsample. When $\lambda = \infty$, $\mbox{WATE}$ is the standard ATE and the asymptotic variance reads 
\begin{equation}
    \mathds{E}\left[\frac{\mathds{V}[Y_i(1)\vert X_i]}{p(X_i)} + \frac{\mathds{V}[Y_i(0)\vert X_i]}{1-p(X_i)}\right] + \mathds{V}[\CATE(X_i)],
\end{equation}
which is the standard efficiency bound. The theorem thus enables inference about a weighted average treatment effect under (unobservable-robust) unconfoundedness in an asymptotically efficient way under some (strong) rate condition on $s_n$, effectively limiting the dimension of continuous unobservables to $5$. \medskip

The weights are tied to the likelihood of belonging to the giant component; individuals whose characteristics make them less likely to form connections are under-represented in the average. There are a few sufficient conditions which ensure that $\mbox{WATE}=\mbox{ATE}$. An obvious one is $\lambda =\infty$, in which case the probability of belonging to the giant component converges to $1$. When $\lambda$ is finite, homogeneous treatment effects or constant covariate density $f$ is a sufficient condition. In general, the weighted average is likely to be close to the average treatment effect if the covariate density is close to uniform, if there are few isolated nodes, or if treatment heterogeneity is limited. \medskip


\subsubsection{Testing for confounding effect of unobservables}

Assuming unconfoundedness conditional on observed covariates $X_i^o$, suppose that a researcher uses observables to construct an estimator of the ATE, say $\widehat{\mbox{ATE}}^o$, whose influence function is the efficient score. \medskip

It may be of interest to test for a possible confounding effect of unobservables featured in network formation by testing whether the difference between $\widehat{\mbox{ATE}}$ and $\widehat{\mbox{ATE}}^o$ is statistically significant.\footnote{I thank an anonymous referee for suggesting this setup for a Hausman test.}\footnote{If one is unwilling to assume one of the sufficient conditions for the WATE to equal the ATE, a possibility is to drop the same individuals from the sample when applying a standard method based on observables and re-interpret the test as testing for equality of WATE.} Formally, if the regularity conditions for asymptotic normality are met, then under the null that the outcome is unaffected by the unobserved components, \textit{i.e.}, $\mbox{CATE}(X_i) \equiv \mbox{CATE}(X_i^o)$, we can establish
\begin{align*}
    \begin{split}
        & \sqrt{n} (\widehat{\mbox{ATE}}-\widehat{\mbox{ATE}}^o) \\ 
        & \overset{d}{\rightarrow} \mathcal{N}\left(0, \mathds{E}\left[\frac{\mathds{V}[Y_i(1)\vert X_i]}{p(X_i)} \frac{(\mathds{P}[T_i=1\vert X_i^o]-p(X_i))^2}{\mathds{P}[T_i=1\vert X_i^o]^2} + \frac{\mathds{V}[Y_i(0)\vert X_i]}{1-p(X_i)} \frac{(\mathds{P}[T_i=1\vert X_i^o]-p(X_i))^2}{(1-\mathds{P}[T_i=1\vert X_i^o])^2}\right]\right)
    \end{split}
\end{align*}
which forms the basis for a test.

\section{Simulations}

I assess the performance of the estimators through simulations. I measure the resulting root mean square error (RMSE) and verify that the method achieves appropriate coverage. The RMSE of standard estimators that use only observed variables is provided for comparison. \medskip

I consider a vector $X$ of $d=7$ covariates where only the first four are observed. The network depends on the last four covariates, so there is some overlap between observables and what influences network. While $X_4$ influences network formation, the fact that it is also observed offers opportunities for further smoothing. \medskip

The propensity score follows a logistic distribution with argument $X_i'\iota$, where $\iota$ is a vector of ones, and the treatment status is then drawn conditional on covariates. \medskip

The outcomes are generated as follows:
\begin{align*}
	&\mu_0(X_i) = 5 \\
    & + \left(\sum_{k=1}^d X_{ik}/d - X_{i4} + X_{i6} + 1/3 (X_{i2}+X_{i3}+X_{i4}+X_{i5}) X_{i6} + {X_{i1}}^2-{X_{i5}}^2\right),\\
	&\mu_1(X_i) = \mu_0(X_i) + 2 \Phi\left(\sum_{k=1}^d X_{ik}\right),\\
	&\varepsilon_i \sim \mathcal{N}(0, 1),\\
	&Y_i = T_i \mu_1(X_i) + (1-T_i) \mu_0(X_i) + \varepsilon_i
\end{align*}

The covariates follow a uniform or a triangular distribution.
The network formation uses $w(x)=e^{-\frac{1}{2} x^2}$ with a baseline $n=500$.
I focus on an exactly sparse regime ($\lambda < \infty$) where $s_n$ is calibrated so that the average number of connections lies between $4$ and $5$, which is common in empirical applications such as friendship networks. \medskip

Network-based weights are constructed based on friendship of order up to $M$, weighting people according to their friendship ranks: friends of order $m$ are weighted by $1/m$ and the number $m$-th order friends.  
Then, the weights $\omega_{ijt}$ are formed by constructing products of Gaussian kernel weights with adaptive bandwidth based on observables and network-based weights, as suggested in remark \ref{kernel_product_remark}. Constructing the estimators of the conditional expectations requires more care to ensure that the convergence conditions are met. Using higher-order kernels to smooth over observables leads to unstable estimates, so local polynomial regressions are generally preferred.
I specify a local cubic regression within each comparison group and select the lowest $M$ such that such regressions are feasible for $1/\log_2(n)$ of the sample, which generally leads to $M=5$ with $n=500$. A regression of $Y_j$ on powers of $(X_j^o-X_i^o)$ up to order three is estimated using observations $j\in \mathcal{C}_{it}$ for each individual in each treatment arm.\footnote{If the number of observations within the comparison group is too low for an individual, the local polynomial regression defaults to a lower order.} \medskip

I compare the performance of the proposed estimator, $\widehat{\ATE}$, to that of AIPW and OLS (using the regression adjusted estimator of the ATE by running separate regressions for both treatment arms).\footnote{From the symmetry of the problem, the WATE coincides with the ATE in this case so the estimates are directly comparable.} The infeasible estimators that use both observed and unobservables as input are also used as an oracle benchmark. \medskip

Results are displayed in Table \ref{sim_results1} and indicate that the estimator provides substantial gains compared to methods that only make use of observables by considerably lowering the bias and improving coverage.

\begin{table}
    \centering
    \begin{tabular}{cccccc}
          Uniform  \\ \hline
          & $\widehat{\ATE}$ & \mbox{AIPW} & \mbox{OLS} & \mbox{infeasible AIPW} & \mbox{infeasible OLS}\\
        \mbox{Bias}& 0.08 & 0.58 & 0.61& -0.01 & 0.02 \\
        \mbox{Standard deviation}& 0.21 & 0.17 & 0.16 & 0.21 & 0.13 \\
         \mbox{RMSE}& 0.22 & 0.61 & 0.63 & 0.21 & 0.14 \\
         \mbox{Coverage}& 0.88 & 0.10 & 0.03 & 0.95 & 0.94 \\ \hline
    \end{tabular}
    \begin{tabular}{cccccc}
         Triangular \\ \hline
        & $\widehat{\ATE}$ & \mbox{AIPW} & \mbox{OLS} & \mbox{infeasible AIPW} & \mbox{infeasible OLS}\\
        \mbox{Bias} & 0.07 & 0.29 & 0.29 & -0.00 & 0.00 \\
        \mbox{Standard deviation} & 0.23 & 0.15 & 0.15 & 0.18 & 0.13 \\
         \mbox{RMSE} & 0.24 & 0.32 & 0.32 & 0.18 & 0.13 \\
         \mbox{Coverage} & 0.88 & 0.55 & 0.44 & 0.95 & 0.93 \\ \hline
    \end{tabular}
    \caption{Simulation results}\label{sim_results1}
\end{table}

To assess the impact of dimensionality, I make network formation depend on 2 additional covariates. Those do not affect the outcomes, so they simply lowers the signal to noise ratio from the network data. 

The performance of the estimator deteriorates: the RMSE increases to $0.43$ (0.29) while coverage decreases to $0.53$ (0.81) in the uniform (triangular) case, but it still provides a noticeable improvement compared to methods that only rely on observables.

\section{Application}

I provide an application of the method to the estimation of the effect of parental involvement on students' test scores. I use the dataset from the project ``Attitudes and Relationships among Primary and High School Students'', see \citet{DVN/ZHCTCK_2022}. The dataset contains information on 4409 Brazilian high-school students, their beliefs, and friendship ties among them. \medskip

The outcome of interest is the average grade, ranging from 0 to 10, and the treatment is the level of parental support. The average is taken over math, Portuguese, English, history, geography, and art grades.\footnote{Although additional subjects are available, including them would require dropping a substantial number of observations because of missing data.} For comparability across school years, I normalize the grade by subtracting the mean over students in a given year. The dataset contains a (self-reported) score for parent support which ranges from around -3 to 1; the treatment is dichotomized by truncating around the mean of 0 and the goal is to estimate the ATE. \medskip

Although some covariates are available--age, gender, race, religion, class dummies, whether parents are employed, neighborhood quality, poverty, and importance of study for the child--they are at best proxies for the underlying causes. Thus, omitted variable bias remains a concern, especially due to the usual unobserved ability. \medskip

The sign of the omitted variable bias is unclear in this context, even assuming that the score for the importance of studying controls adequately for motivation. Intuitively, children with lower ability may require more attention but are also more likely through heredity to have less able parents, who may be less inclined or able to help. \medskip

Formally, a possible model for the individual's grade, $y_i$, would be $y_i = F(a_i; m_i; s_i)$ where $F$ is an unknown function, $a_i$ is ability, $m_i$ is the level of (intrinsic) motivation,\footnote{Grades may be affected directly or indirectly through increased effort. The target is the total effect that includes the mediated effect, conditioning on intrinsic motivation and ability.} and $s_i$ is parental support. Denoting parents' ability by $A_i$, one could specify $a_i = A_i+\mbox{noise}_i$ and $s_i = \delta_1 a_i + \delta_2 A_i + \mbox{noise}_i$ with independent noises. This would imply $s_i = (\delta_1+\delta_2) a_i + \mbox{noise}_i$, where the signs of the deltas are likely to be negative and positive, respectively, leading to an indeterminate correlation sign. Hence, it is not obvious whether parents pay more attention to children with higher or lower ability. Moreover, the function $F$ is unlikely to be linear since, \textit{e.g.}, ability may increase the return to motivation and parent support and there may be nonlinear returns to motivation and/or ability. \medskip 

As a result, not only is the OLS estimate likely to be unreliable, but it is also difficult to figure out the direction of the bias. This motivates the use of alternative methods that can account for the presence of unobservables. \medskip

There is evidence that people, teenagers in particular, form friendship links based on intelligence. Some studies \citep{clark1992friendship, burgess2011school, boutwell2017general} have documented homophilic matching on various measures of intelligence among teenagers. According to \citet{boutwell2017general}, ``preadolescent friendship dyads are robustly correlated on measures of general intelligence''. \medskip


It is thus plausible that friendship ties account for ability, suggesting an avenue for correcting the ability bias through the estimator developed in this paper.
First, I extract the largest component within each classroom and keep the components with at least $5$ individuals from each treatment arms. Then, I form comparison groups based on $M=3$ for each individual and control for observables with kernel weights and regression adjustments, similarly to the simulation design. \medskip

According to OLS,\footnote{OLS refers here to estimated ATE obtained as the mean difference in fitted values between the treatment-arm-specific regressions. A simple regression of $y$ on $T$ and covariates delivers similar results.} the treatment has a small effect of 0.06. The estimator developed in the paper, however, suggests a much higher, significant effect of $0.33$. Because of the control for unobserved ability and its potential nonlinear interactions, the latter estimates may be more reasonable. Since the estimator targets a weighted ATE, another interesting comparison comes from computing the regression-based estimate on the subsample consisting of giant components. The resulting estimate, $\widehat{\mbox{WATE}}_{OLS}$, is larger (0.15), but still not statistically significant (se 0.14). 

\begin{table}[!ht]\caption{Estimates and standard errors}\label{application_results}
\[\begin{array}{ | l | l | l |}
\hline
 & \mbox{Estimated ATE} & \mbox{standard error}  \\ \hline
\widehat{\mbox{WATE}} & 0.33 & 0.13  \\ \hline
\widehat{\mbox{ATE}}_{OLS} & 0.06 & 0.05  \\ \hline
\widehat{\mbox{WATE}}_{OLS} & 0.15 & 0.14  \\ \hline
\end{array}\]
\end{table}

\begin{remark}
    Running a regression without the parent's employment status and the poverty score increases the treatment effect estimate from OLS to 0.09. Going back to the model $s_i = \delta_1 a_i + \delta_2 \mbox{parents' ability}_i + \mbox{noise}_i$, this may indicate that OLS may be \textbf{more} biased upon controlling for parent's employment and poverty because these variables may act as proxies for parent ability and thus increase the conditional correlation between ability and parent support. \medskip
\end{remark}

\bibliographystyle{aea}
\bibliography{References}

\section*{Appendix A: Proofs}

\subsection{Lemma \ref{lemma}}

\begin{proof}
Consider a sample of $(n+1)$ observations, including $i$. By independence, the degree of individual $i$ satisfies $d_i \equaldef \vert \mathcal{N}(i) \vert = \sum_{j \neq i, j=1}^{n+1} \mathds{1}_{\eta_{ij} \leq w_n(\Vert x_i - X_j \Vert)} \sim \mathcal{B}(n, \pi_n)$ with $\pi_n \equaldef \mathds{E}[w_n(\Vert x_i - X_j\Vert)] = \int_{\mathds{R}^d} w_n(\Vert x_j-x_i\Vert) f(x_j) \ \mbox{d}x_j$ by the law of iterated expectation and distributional properties of $\eta$. If $n \pi_n \rightarrow C > 0$, the friend count asymptotically follows a Poisson distribution with parameter $C$, while slower sequences $\pi_n$ induce an unbounded friend count: by Chebyshev inequality we have for any $N$ and large sufficiently large $n$
\begin{align*}
    \begin{split}
        \mathds{P}[d_i \leq N] & \leq \mathds{P}[\vert d_i - n \pi_n\vert \geq n \pi_n - N] \\
        & \leq \frac{n \pi_n (1-\pi_n)}{(n \pi_n - N)^2} \\
        & = \frac{1-\pi_n}{n \pi_n (1-\frac{N}{n \pi_n})^2}
    \end{split}
\end{align*}
and thus $\mathds{P}[d_i \leq N] \rightarrow 0$ as long as $n \pi_n = n \int_{\mathds{R}^d} w_n(\Vert x_j-x_i\Vert) f(x_j) \ \mbox{d}x_j \rightarrow \infty$. \medskip

This is the case when the network formation is regularly asymptotically homophilic. Indeed, consider

\begin{align*}
    \begin{split}
        \int_{\mathds{R}^d} w_n(\Vert x-x_i \Vert) f(x) \ dx & = \int_{\mathds{R}^d} g(s_n \Vert x-x_i \Vert) f(x) \ dx  \\
        & \geq \underline{f} \int_{B_{\underline{r}}(x_i)} g(s_n \Vert x-x_i \Vert) \ dx \\
        & \geq \frac{\underline{f}}{{s_n}^{d}} \int_{B_{s_n \underline{r}}(0)} g(\Vert y\Vert) \ dy \\
        & \geq \frac{\underline{f}}{{s_n}^{d}} \int_{B_{\underline{r}}(0)} g(\Vert y\Vert) \ dy
    \end{split}
\end{align*}
for sufficiently large $n$. Since $\int_{\mathds{R}^d} g(\Vert y\Vert) \ dy > 0$ and $g$ is decreasing, $\int_{B_{\underline{r}}(0)} g(\Vert y\Vert) \ dy > 0$ (otherwise, $\int_{B_{\underline{r}s_n}(0)} g(\Vert y\Vert) \ dy = 0$ implies $g \equiv 0$, contradicting the former assertion). Thus, \textit{regular asymptotic homophily} suffices to establish $\lim_{n \rightarrow \infty} n \int_{\mathds{R}^d} w_n(\Vert x_j-x_i\Vert) f(x_j) \ \mbox{d}x_j = \infty$ when $\lambda = \infty$. \medskip

For part a), write $w_n(\Vert x_i - x_j\Vert)$ as $w_{ij}$.
Noting that the covariate density is continuous and bounded, and that $g$ is bounded, the dominated convergence theorem applies below and thus
 \begin{align*}
    \begin{split}
	\mathds{E}[\vert \mathcal{N}_i^m\vert] & = \frac{n!}{(n-m)!} \int \prod_{k=0}^{m-1} g(s_n \Vert x_{j_k} - x_{j_{k+1}}\Vert) \prod_{k=1}^m f(x_{j_k}) \prod_{l> k + 1} (1-w_{j_k j_l}) \ \mbox{d}\left(\prod_{k=1}^m x_{j_k}\right) \\
	& = O(n^m) {s_n}^{-md} \int \prod_{k=1}^{m} g(\Vert y_k\Vert) \prod_{k=1}^m f\left(x_i - \frac{y_k+y_{k-1}+...}{s_n}\right) \ \mbox{d}\left(\prod_{k=1}^m y_{k}\right) \\
	& = O(n^m) \frac{\psi_i^m}{{s_n}^{md}} O(1)
\end{split}
\end{align*}
which establishes that $\mathds{E}[\vert \mathcal{N}_i^m\vert] = O\left((n \psi_i {s_n}^{-d})^m\right)$, and thus collecting friendship up to the $M$-th order gives $\mathds{E}[\vert \mathcal{C}_i \vert] = O\left((n \psi_i {s_n}^{-d})^M\right)$. 




Finally, for part b), the number of friends in common between $i$ and $j$ follows a binomial distribution with probability $O({s_n}^{-2d})$. Since probability of having $c$ or more friends concentrates on exactly $c$ friends, and that there are $O(n^c)$ ways of picking those $c$ individuals, the order of the probability is $O(n^c) O(({s_n}^{-2d})^c) = O((n {s_n}^{-2d})^c)$.

\end{proof}




\subsection{Theorem \ref{CATE_consistency}}

\begin{proof}
\textbf{Part a)}
Lemma 2.1 ensures that the asymptotics apply under the relevant condition, possibly using the higher-order friendship with growing $M$ if $\lambda < \infty$. By the overlap assumption, the number of treated friends and untreated friends both grow to infinity. 

Consider the treated group. Letting $\tilde{Y}_{j;n}$ denote i.i.d. draws from $Y(1)\vert \mathcal{C}_{i1}$, we have $\sup_n \mathds{E}[(\tilde{Y}_{j;n})^2] < \infty$ since $\mathds{E}[Y_j(1)^2 \vert \mathcal{C}_{i1}] \rightarrow \mathds{E}[Y_j(1)^2 \vert X_j = x_i, T_j = 1] = \mathds{E}[Y_j(1)^2 \vert X_j = x_i]$ by continuity (this is shown in details for $Y_j(1)$ below) and unconfoundedness. 
Then, $\frac{1}{\vert \mathcal{C}_{i1}\vert} \sum_{j \in \mathcal{C}_{i1}} (Y_j(T_j) - \mathds{E}[Y_j(1)|\mathcal{C}_{i1}]) + \mathds{E}[Y_j(1)|\mathcal{C}_{i1}] - \mathds{E}[Y_j(1)|X_j=x_i] \overset{p}{\rightarrow} 0$ by the law of large numbers for triangular arrays and continuity, which is formally established by bounding $\mathds{E}[Y_j(1)|\mathcal{C}_{i1}] - \mathds{E}[Y_j(1)|X_j=x_i]$ as follows. 

Divide the expectation into an integral over a $\varepsilon$-ball centered at $x_i$ and an integral over its complement. Within the ball, 
\begin{align*}
    \begin{split}
        & \left\vert \int_{B_\varepsilon(x_i)} (\mathds{E}[Y_j(1)|X=x]-\mathds{E}[Y_j(1)|X=x_i]) f_{X|\mathcal{C}_i, T} \ dx \right\vert\\
        & \leq \sup_{B_\varepsilon} \vert \mathds{E}[Y_j(1)|X=x]-\mathds{E}[Y_j(1)|X=x_i] \vert \\
        & = O(\varepsilon^\alpha)
    \end{split}
\end{align*}
by Hölder continuity. \medskip

The integral outside the ball, \\
$\int_{B_\varepsilon^c(x_i)} (\mathds{E}[Y_j(1)|X=x]-\mathds{E}[Y_j(1)|X=x_i]) f_{X\vert \mathcal{C}_i, T} \ dx$, can be bounded using information about the composition of $\mathcal{C}_i$. I proceed with the two groups separately. \\
For the first choice of comparison group, \textit{i.e.}, $\mathcal{C}_i=\cup_{m=1}^M \mathcal{N}^m(i)$, note that $\varepsilon < \Vert x_j-x_i\Vert = \left\Vert \sum_{k=0}^{m-1} x_{j_{k+1}} - x_{j_k}\right\Vert \leq \sum_{k=0}^{m-1} \Vert x_{j_{k+1}} - x_{j_k}\Vert$ and thus outside the ball

\begin{align*}
    \begin{split}
        f_{\mathcal{C}_i\vert X_j, T_j} & = \sum_{m=1}^M \mathds{P}[j \in \mathcal{N}_i^m, j \notin \mathcal{N}_i^{m'}, m'<m\vert X, T] \\
        & = \sum_{m=1}^M \mathds{E}[\mathds{P}[j \in \mathcal{N}_i^m, j \notin \mathcal{N}_i^{m'}, m'<m\vert \{X_{j_k}\}, T]\vert X_j, T_j]\\
        & \leq \sum_{m=1}^M \int \prod_{k=0}^{m-1} w_{j_{k+1} j_k} \prod_{l \neq k} (1-w_{j_k j_l}) f_{X_{-j}\vert X_j} \\
        & \leq \sum_{m=1}^M w_n\left(\frac{\varepsilon}{m}\right)^m
    \end{split}
\end{align*}

\noindent It follows that 

\begin{align*}
    \begin{split}
        & \left\vert \int_{B_\varepsilon^c(x_i)} (\mathds{E}[Y_j(1)|X=x]-\mathds{E}[Y_j(1)|X=x_i]) f_{X\vert \mathcal{C}_i, T} \ dx \right\vert \\
        & \leq \int_{B_\varepsilon^c(x_i)} \left\vert \mathds{E}[Y_j(1)|X=x]-\mathds{E}[Y_j(1)|X=x_i]\right\vert \frac{\mathds{P}[\mathcal{C}_{i} \vert X=x] f_{X, T}(x, 1)}{\mathds{P}[\mathcal{C}_i, T=1]} \ dx \\
        & \leq \frac{\sum_{m=1}^M \left(w_n(\frac{\varepsilon}{m})\right)^m (\mathds{E}[\vert \mathds{E}[Y_j(1)|X]\vert]+\mathds{E}[Y_j(1)|X=x_i])}{\mathds{P}[\mathcal{C}_i, T=1]} \\
        & \leq C \frac{n}{{\lambda_n}^M} \sum_{m=1}^M \left(w_n\left(\frac{\varepsilon}{m}\right)\right)^m
    \end{split}
\end{align*}

Hence, the integrals are $O({\varepsilon_n}^{\alpha})$ and $O\left(\frac{n}{{\lambda_n}^M} \sum_{m=1}^M (w_n(\frac{\varepsilon_n}{m}))^m\right)$, respectively, and the bias disappears provided $\sum_{m=1}^M (w_n(\frac{\varepsilon_n}{m}))^m = o({\lambda_n}^M n^{-1})$. 
Applying the same reasoning to the non-treated gives the result for the first choice of comparison group. \medskip

For the second comparison group, the conditional distribution (leaving the conditioning on $T=t$ implicit) satisfies 

\begin{align*}
    \begin{split}
        f_{X_j \vert \mathcal{C}_{i}} & = \frac{\mathds{P}[\mathcal{C}_{i}\vert X_j] f(x)}{\int \mathds{P}[\mathcal{C}_{i}\vert X_j] f(x) \ \mbox{d}x} \\
        & = \frac{\mathds{E}[\mathds{P}[W_{i1} = W_{j1} = \cdots = W_{ic} = W_{jc} = 1\vert X_1, \cdots, X_j]] f(x)}{\int \mathds{E}[\mathds{P}[W_{i1} = W_{j1} = \cdots = W_{ic} = W_{jc} = 1\vert X_1, \cdots, X_j]] f(x) \ \mbox{d}x} \\
        & = \frac{\prod_{\tilde{c}=1}^c \int w_n(\Vert x_i - x_k\Vert) w_n(\Vert x_j - x_k\Vert) f(x_k) \mbox{d}x_k f(x)}{\prod_{\tilde{c}=1}^c \int w_n(\Vert x_i - x_k\Vert) w_n(\Vert x_j - x_k\Vert) f(x_k) \mbox{d}x_k f(x) \ \mbox{d}x} \\
        & \leq C {s_n}^d \sum_{\tilde{c}=c}^{n-2} \left(w_n\left(\frac{\varepsilon_n}{2}\right)\right)^c \\
        & = C {s_n}^d \left(w_n\left(\frac{\varepsilon_n}{2}\right)\right)^c \ O(1)
    \end{split}
\end{align*}
so that we want $\left(w_n(\frac{\varepsilon_n}{2})\right)^c = o({s_n}^{-d})$

\textbf{Part b)}
One can compute the bias originating from each exact friendship order $m$, then aggregate in a way similar to the previous arguments. For $\mathcal{C}_i$ consisting of $m$-th order friends, we have
\begin{align*}
    \begin{split}
       \mathds{E}[Y_j(t) \vert \mathcal{C}_i, T_j=t] & = \mathds{E}[\mathds{E}[Y_j(t) \vert X_j]\vert \mathcal{C}_i, T_j=t] \\
        & = \int_{\mathds{R}^d} \mathds{E}[Y_j(t) \vert x_j] f_{x_j\vert \mathcal{C}_i, x_i, t} \mbox{d}x_j \\
        & = \int_{\mathds{R}^d} \mathds{E}[Y_j(t) \vert x_j] f_{x_j\vert \mathcal{C}_i, x_i, t}  \mbox{d}x_j \\
        & = \int_{\mathds{R}^d} \mathds{E}[Y_j(t) \vert x_j] \frac{\mathds{P}[j \in \mathcal{C}_{i}\vert x_i, x_j] p_{t}(x_j) f_{x_j}}{\int_{\mathds{R}^d} \mathds{P}[j \in \mathcal{C}_{i}\vert x_i, x_j] p_{t}(x_j) f_{x_j} \mbox{d}x_j} \mbox{d}x_j
    \end{split}
\end{align*}
Expanding the probability of $\mathds{P}[j \in \mathcal{C}_i]$ into possible paths $x_i, x_{j_1}, ... , x_{j_{m-1}}, x_j$ from $i$ to $j$ as in the previous proofs, and introducing the changes of variable,
\begin{equation*}
    \begin{pmatrix}
    x_i \\ x_{j_1}-x_i \\ ... \\ x_j - x_{j_{m-1}}
\end{pmatrix} = 
\begin{pmatrix}
    1 & 0 & 0 & \ldots & 0 \\ 1 & -1/s_n & 0 & \ldots& 0 \\ 1 & -1/s_n & -1/s_n & \ldots & 0 \\ \ldots & \ldots & \ldots & \ldots & \ldots \\ 1 & -1/s_n & -1/s_n & \ldots & -1/s_n
\end{pmatrix} \begin{pmatrix}
    x_i \\ y_1 \\ \ldots \\ y_m
\end{pmatrix},
\end{equation*}
with Jacobian determinant $\vert (-1)^m {s_n}^{-md}\vert$, and letting $y^{\tilde{m}} \equaldef \sum_{k=1}^{\tilde{m}} y_k$, we obtain 
 \begin{align*}
    \begin{split}       
        & = \int_{(\mathds{R}^d)^m} \mathds{E}\left[Y_j(t) \vert x_i-y^m/s_n\right] \frac{\prod_{{\tilde{m}}=1}^m g(\Vert y_{\tilde{m}} \Vert) p_{t}(x_i-y^m/s_n) \prod_{{\tilde{m}}=1}^m f_{x_i-y^{\tilde{m}}/s_n}}{\int_{(\mathds{R}^d)^m} \prod_{{\tilde{m}}=1}^m g(\Vert y_{\tilde{m}} \Vert) p_{t}(x_i-y^m/s_n) \prod_{{\tilde{m}}=1}^m f_{x_i-y^{\tilde{m}}/s_n} \prod_{{\tilde{m}}=1}^m \mbox{d}y_{\tilde{m}}} \prod_{{\tilde{m}}=1}^m \mbox{d}y_{\tilde{m}} \\
        & = \mathds{E}[Y_j(t) \vert x_i] + C {s_n}^{-2} m \int \lVert y\rVert^2 g(\lVert y\rVert) \mbox{d}y \frac{(\int g(\lVert y\rVert) \mbox{d}y)^{m-1}}{(\int g(\lVert y\rVert) \mbox{d}y)^m}\\
        & = \mathds{E}[Y_j(t) \vert x_i] + O(m {s_n}^{-2})
    \end{split}
\end{align*}
using a second-order expansion in $y$ in all densities, expectation, and propensity scores (noting that $\frac{A+a_n}{B+b_n}=\frac{A}{B} \left(1+\frac{a_n}{A}\right)\left(1-\frac{b_n}{B+b_n}\right)=\frac{A}{B} + O(a_n) + O(b_n) + O(a_n b_n)$ and that the $O({s_n}^{-1})$ terms vanish by symmetry: $\int y_{\tilde{m}} g(\lVert y_{\tilde{m}} \rVert) \mbox{d}y_{\tilde{m}}=0 $) and that second-order derivatives are bounded.

Aggregating over $m$ finally yields the order of $O(M^2 {s_n}^{-2})$.

\end{proof}

\subsection{Theorem \ref{CATE_asymptotic_normality}}

\begin{proof}

By the Law of Large Numbers,
\begin{equation}
    \left(\frac{\vert \mathcal{C}_{it}\vert}{\mathds{E}[\vert \mathcal{C}_{it}\vert]} - 1\right) = \frac{1}{n} \sumjn \left(\frac{\mathds{1}(j \in \mathcal{C}_{it})}{\mathds{P}[j \in \mathcal{C}_{it}]} - 1\right) \overset{p}{\rightarrow} 0
\end{equation}
for $t \in \{0, 1\}$.\footnote{If $M\neq 1$, there is some correlation between the indicators through the chain of friends. However, this correlation is negligible since $M$ remains fixed or scales slower than even $\ln(n)$.} Next, by Lindeberg-Feller's central limit theorem where Lindeberg's condition follows from the dominated convergence theorem under $2+\delta$-moment existence through Lyapunov's condition,
\begin{align*}
    \begin{split}
        & \begin{pmatrix} \frac{1}{\sqrt{n}} \sumjn \frac{\mathds{1}(j \in \mathcal{C}_{i1})}{\sqrt{\mathds{P}[j \in \mathcal{C}_{i1}]}} \left(Y_j - \mathds{E}[Y_j(1)|\mathcal{C}_i, T_j=1]\right) \\ \frac{1}{\sqrt{n}} \sumjn \frac{\mathds{1}(j \in \mathcal{C}_{i0})}{\sqrt{\mathds{P}[j \in \mathcal{C}_{i0}]}} \left(Y_j - \mathds{E}[Y_j(0)|\mathcal{C}_i, T_j=0]\right) \end{pmatrix}\\
        & \overset{d}{\rightarrow} \mathcal{N}\left(\begin{pmatrix} 0 \\ 0 \end{pmatrix}; \begin{pmatrix}\mathds{V}[Y_j(1)|X_j=x_i] & 0 \\ 0 & \mathds{V}[Y_j(0)|X_j=x_i] \end{pmatrix}\right)
    \end{split}
\end{align*}
Combining these two results and noting that $\vert \mathcal{C}_{i} \vert = \vert \mathcal{C}_{i1} \vert + \vert \mathcal{C}_{i0} \vert$, it follows that 
\begin{equation}
    \sqrt{\mathds{E}[\vert \mathcal{C}_i\vert]} (\widehat{\CATE}(x_i; \mathcal{C}_i)-\CATE(x_i)-\mathds{B}_i) \overset{d}{\rightarrow} \mathcal{N}\left(0; V\right)
\end{equation}

Finally, the bias is negligible when $\sqrt{\mathds{E}[\vert \mathcal{C}_i\vert]}  (\mathds{E}[Y_j(t)\vert \mathcal{C}_i, T_j=t] - \mathds{E}[Y_j(t)|X_j=x_i]) \rightarrow 0$. From the consistency proof, it is seen that this occurs if $\varepsilon_n = n^{-\gamma}$ under the conditions of the theorem.
\end{proof}

\subsection{Theorem \ref{ATE_CAN}}

In what follows, the sample is implicitly restricted to observations belonging to the giant component. Averaging and expectations are thus understood as conditional on $\mathcal{G}$ so that, \textit{e.g.}, mentions of the ATE refer to the WATE.

The main preliminary result, which consists in establishing a variant of Assumption 3.3 (ii) in \citet{lin2025regression}, is the mean-square convergence of the weights to inverse-propensity score quantities:
\begin{lemma}\label{weights_to_score_lemma}
    We have 
    \begin{equation}
    \mathds{E}\left[\left(\sum_j \omega_{ji T_i} - \frac{T_i}{p(X_i)} - \frac{1-T_i}{1-p(X_i)}\right)^2\right] \rightarrow 0
    \end{equation}
\end{lemma}
\begin{proof}

The proof follows steps similar to Lemma B.3. in \citet{lin2025regression}. In what follows, $\boldsymbol{T}$ denotes the vector of treatment statuses and $\boldsymbol{X}$ the collection of covariates, $f_t$ is the density of $X$ conditional on treatment status $t \in \{0, 1\}$, and it is understood that $j \neq i$ when conditioning on both $X_i$ and $X_j$. Note that $j \in \mathcal{C}_i$ iff $i \in \mathcal{C}_j$ for the comparison groups of interest. To keep notation simple, the proof is presented for $M=1$.

First, decompose 
\begin{align*}
    \begin{split}
        \mathds{E}\left[\left(\sum_j \omega_{ji T_i} - \frac{T_i}{p(X_i)} - \frac{1-T_i}{1-p(X_i)}\right)^2\right] & = \mathds{E}\left[\mathds{E}\left[\left(\sum_j \omega_{ji 1} - \frac{1}{p(X_i)}\right)^2 \middle\vert T, T_i = 1\right] T_i \right] \\
        & + \mathds{E}\left[\mathds{E}\left[\left(\sum_j \omega_{ji 0} - \frac{1}{1-p(X_i)}\right)^2 \middle\vert T, T_i = 0\right] (1-T_i) \right]
    \end{split}
\end{align*}

The two terms can be handled analogously, so it suffices to consider the case where $i$ is treated. The corresponding term can be decomposed further into
\begin{equation}
    \sum_j \omega_{ji1} - \frac{1}{p(X_i)} = \frac{n (R_4+R_3+R_2)}{n_1} + \frac{n}{n_1} \frac{f(X_i)}{f_1(X_i)} - \frac{1}{p(X_i)}
\end{equation}
where 
\begin{equation}
    R_2 \equaldef f_1^{-1}(X_i) \left[\frac{1}{n} \sum_{j} \frac{\mathds{1}(j \in \mathcal{C}_i) {s_n}^d}{\int g(\Vert y\Vert) \ \mbox{d}y} - f(X_i)\right]
\end{equation}
\begin{equation}
    R_3 \equaldef \frac{1}{n} \sum_{j} (f_1^{-1}(X_j)-f_1^{-1}(X_i)) \ \frac{\mathds{1}(j \in \mathcal{C}_i) {s_n}^d}{\int g(\Vert y\Vert) \ \mbox{d}y}
\end{equation}
\begin{equation}
    R_4 \equaldef \frac{1}{n} \sum_{j} \left[\left(\frac{1}{n_1} \sum_{k: T_k=1} \frac{\mathds{1}(k \in \mathcal{C}_j) {s_n}^d}{\int g(\Vert y\Vert) \ \mbox{d}y} \right)^{-1} - f_1^{-1}(X_j) \right]  \frac{\mathds{1}(j \in \mathcal{C}_i) {s_n}^d}{\int g(\Vert y\Vert) \ \mbox{d}y}
\end{equation}

\noindent From the law of large numbers, 
\begin{equation}
    \mathds{E}\left[\mathds{E}\left[ \left(\frac{n}{n_1} \frac{f(X_i)}{f_1(X_i)} - \frac{1}{p(X_i)}\right)^2 \middle\vert \boldsymbol{T}, T_i=1\right] T_i\right] \rightarrow 0
\end{equation}

\noindent So it remains to show that
\begin{equation}
    \mathds{E}\left[\mathds{E}\left[ \left(\frac{n R_m}{n_1}\right)^2 \middle\vert \boldsymbol{T}, T_i=1\right] T_i\right] \rightarrow 0
\end{equation}
for $m = 2, 3, 4$. 

\subsection*{$R_2$}

Using iterated expectations to condition on $X_i$, the second moment can be decomposed into the bias

\begin{align*}
& \mathds{E}\left[\left\{ \mathds{E}\left[ \frac{\mathds{1}(j\in \mathcal{C}_i) {s_n}^d}{\int g(\Vert y\Vert) \ \mbox{d}y} \middle\vert X_i, T_i=1\right] - f(X_i) \right\}^2 \right]\\
& = \mathds{E}\left[\left\{ \mathds{E}\left[\mathds{E}\left[\frac{\mathds{1}(j\in \mathcal{C}_i) {s_n}^d}{\int g(\Vert y\Vert) \ \mbox{d}y} \middle\vert X_i, X_j, T_i=1\right] \middle\vert X_i, T_i=1\right] - f(X_i) \right\}^2 \right] \\
& = \mathds{E}\left[\left\{ \int \frac{g(s_n \Vert x_j - X_i\Vert)}{{s_n}^{-d} \int g(\Vert y \Vert) \ \mbox{d}y} f(x_j) \ \mbox{d}x_j - f(X_i) \right\}^2 \right] \\
& = \mathds{E}\left[\left\{ \int \frac{g(\Vert z\Vert)}{\int g(\Vert y \Vert) \ \mbox{d}y} (f(X_i+z/s_n) - f(X_i)) \ \mbox{d}z \right\}^2 \right] \\
& = O({s_n}^{-4})
\end{align*}
and the variance
\begin{align*}
& \mathds{E}\left[\frac{1}{n} \mathds{V} \left[ \frac{\mathds{1}(j\in \mathcal{C}_i) {s_n}^d}{\int g(\Vert y\Vert) \ \mbox{d}y} \middle\vert X_i, T_i=1 \right] \right]\\
&\leq \mathds{E}\left[\frac{1}{n} {s_n}^{2d} \mathds{E}\left[ \left(\frac{\mathds{1}(j \in \mathcal{C}_i)}{\int g(\Vert y\Vert) \mbox{d}y}\right)^2 \middle\vert X_i, T_i=1 \right]\right]\\
&\leq \frac{1}{n} {s_n}^{2d} \frac{O(1)}{{s_n}^{d}}\\
&\leq O((n{s_n}^{-d})^{-1})
\end{align*}
which implies 
\begin{equation}
    \mathds{E}\left[\mathds{E}\left[ \left(\frac{n R_2}{n_1}\right)^2 \middle\vert T, T_i=1\right] T_i\right] \rightarrow 0
\end{equation}

\subsection*{$R_3$}

After iterated expectations, the bias equals (the expectation of)
\begin{align*}
& \left\{ \mathds{E}\left[ (f_1^{-1}(X_j)-f_1^{-1}(X_i)) \ \frac{\mathds{1}(j \in \mathcal{C}_i) {s_n}^d}{\int g(\Vert y\Vert) \ \mbox{d}y} \middle\vert X_i, T_i=1, T_j=1\right]\right\}^2 \\
& = \left\{ \int (f_1^{-1}(x_j)-f_1^{-1}(X_i)) \ \frac{g(s_n \Vert x_j - X_i\Vert) {s_n}^d}{\int g(\Vert y\Vert) \ \mbox{d}y} \ f_1(x_j) \mbox{d}x_j \right\}^2 \\
& = \iint (f_1^{-1}(X_i + z/s_n)-f_1^{-1}(X_i)) (f_1^{-1}(X_i + \tilde{z}/s_n)-f_1^{-1}(X_i)) \ \frac{g(\Vert z \Vert) g(\Vert \tilde{z} \Vert)}{(\int g(\Vert y\Vert) \ \mbox{d}y)^2} \\
& \times f_1(X_i + z/s_n) f_1(X_i + \tilde{z}/s_n) \ \mbox{d}z \mbox{d}\tilde{z} \rightarrow 0
\end{align*}
while the variance term corresponds to
\begin{align*}
& \frac{1}{n} \mathds{V}\left[ (f_1^{-1}(X_j)-f_1^{-1}(X_i)) \ \frac{\mathds{1}(j \in \mathcal{C}_i) {s_n}^d}{\int g(\Vert y\Vert) \ \mbox{d}y} \middle\vert X_i, T_i=1, T_j=1\right] \\
& \leq \frac{1}{n {s_n}^{-2d}} \mathds{E}\left[\frac{(f_1^{-1}(X_j)-f_1^{-1}(X_i))^2 \ \mathds{1}(j \in \mathcal{C}_i)}{(\int g(\Vert y\Vert) \ \mbox{d}y)^2}\right] \rightarrow 0
\end{align*}

\subsection*{$R_4$}

$R_4$ can be rewritten as $\frac{1}{n} \sum_{j} f_1^{-2}(X_j) \left[f_1(X_j) - \frac{1}{n_1} \sum_{k\neq j, T_k=1} \frac{\mathds{1}(k \in \mathcal{C}_j) {s_n}^d}{\int g(\Vert y\Vert) \ \mbox{d}y} \right] \frac{\mathds{1}(j \in \mathcal{C}_i) {s_n}^d}{\int g(\Vert y\Vert) \ \mbox{d}y}$,
up to lower-order terms. Then, by Jensen's inequality and iterated expectations,
\begin{align*}
\begin{split}
    & \mathds{E}\left[\left(\frac{1}{n} \sum_{j} f_1^{-2}(X_j) \left[f_1(X_j) - \frac{1}{n_1} \sum_{k\neq j, T_k=1} \frac{\mathds{1}(k \in \mathcal{C}_j) {s_n}^d}{\int g(\Vert y\Vert) \ \mbox{d}y} \right] \frac{\mathds{1}(j \in \mathcal{C}_i) {s_n}^d}{\int g(\Vert y\Vert) \ \mbox{d}y}\right)^2\right] \\ & \leq \mathds{E}\left[\frac{1}{n} \sum_{j} \left(f_1^{-2}(X_j) \left[f_1(X_j) - \frac{1}{n_1} \sum_{k\neq j, T_k=1} \frac{\mathds{1}(k \in \mathcal{C}_j) {s_n}^d}{\int g(\Vert y\Vert) \ \mbox{d}y} \right] \frac{\mathds{1}(j \in \mathcal{C}_i) {s_n}^d}{\int g(\Vert y\Vert) \ \mbox{d}y}\right)^2\right] \\
    & \leq \mathds{E}\left[\left(f_1^{-2}(X_j) \left[f_1(X_j) - \frac{1}{n_1} \sum_{k\neq j, T_k=1} \frac{\mathds{1}(k \in \mathcal{C}_j) {s_n}^d}{\int g(\Vert y\Vert) \ \mbox{d}y} \right] \frac{\mathds{1}(j \in \mathcal{C}_i) {s_n}^d}{\int g(\Vert y\Vert) \ \mbox{d}y}\right)^2\right] \\
    & \leq \mathds{E}\left[\mathds{E}\left[\left(f_1^{-2}(X_j) \left[f_1(X_j) - \frac{1}{n_1} \sum_{k\neq j, T_k=1} \frac{\mathds{1}(k \in \mathcal{C}_j) {s_n}^d}{\int g(\Vert y\Vert) \ \mbox{d}y} \right] \frac{\mathds{1}(j \in \mathcal{C}_i) {s_n}^d}{\int g(\Vert y\Vert) \ \mbox{d}y}\right)^2\middle\vert X_i, X_j\right]\right].
\end{split}
\end{align*}
Decomposing the expectation of the square into squared bias and variance, and proceeding as for $R_2$ and $R_3$ then implies $\mathds{E}\left[\mathds{E}\left[ \left(\frac{n R_4}{n_1}\right)^2 \middle\vert \boldsymbol{T}, T_i=1\right] T_i\right] \rightarrow 0$.
\end{proof}

The rest of the proof follows arguments similar to those in  \citet{lin2025regression}, who establish the asymptotic normality of ATE imputation estimators that use kernel weights. First, letting $\varepsilon_i \equaldef Y_i - \mu_{T_i}(X_i)$, note that the estimator can be written as
\begin{align*}
\begin{split}
    \widehat{\mbox{ATE}} &= \meanin
    \left[ \sumjn T_j \omega_{ij1} \left(Y_j - \hat{\mu}_1(X_j) + \hat{\mu}_1(X_i)\right)
    - \sumjn (1-T_j) \omega_{ij0} \left(Y_j - \hat{\mu}_0(X_j) + \hat{\mu}_0(X_i)\right)
    \right] \\
    &= \meanin
    \Bigg[ \sumjn T_j \omega_{ij1} \left(Y_j - \hat{\mu}_1(X_j) + \hat{\mu}_1(X_i) +\mu_1(X_i)-\mu_1(X_j)-\mu_1(X_i)+\mu_1(X_j) \right)\\
    &- \sumjn (1-T_j) \omega_{ij0} \left(Y_j - \hat{\mu}_0(X_j) + \hat{\mu}_0(X_i) + \mu_0(X_i)-\mu_0(X_j)-\mu_0(X_i)+\mu_0(X_j)\right)
    \Bigg] \\
    & = \meanin \left( \mu_1(X_i)-\mu_0(X_i) \right) \\
    & + \meanin \Bigg[
    \sum_j T_j \omega_{ij1} \left( \varepsilon_j - \mu_1(X_i)+\mu_1(X_j) + \hat{\mu}_1(X_i)-\hat{\mu}_1(X_j) \right) \\
    & - \sumjn (1-T_j) \omega_{ij0}
    \left(\varepsilon_j - \mu_0(X_i)+\mu_0(X_j)+\hat{\mu}_0(X_i)-\hat{\mu}_0(X_j) \right)
    \Bigg] \\
    & = S + \tilde{R},
\end{split}
\end{align*}
where  
\begin{align*}
    S & = \meanin \mu_1(X_i)-\mu_0(X_i) + \meanin \left(\sumjn T_j \omega_{ij1} \varepsilon_j - \sumjn (1-T_j) \omega_{ij0} \varepsilon_j\right) \\
    & = \meanin \mu_1(X_i)-\mu_0(X_i) + \meanin \sumjn (2T_j - 1) \omega_{ijT_j} \varepsilon_j \\
    & = \meanin \mu_1(X_i)-\mu_0(X_i) + \meanin (2T_i-1) \sumjn \omega_{ji T_i} \varepsilon_i
\end{align*}
and
\begin{align*}
    \tilde{R} & = \meanin 
    \Bigg( \sumjn T_j \omega_{ij1} \left(\hat{\mu}_{1}(X_i)-\mu_{1}(X_i)+\mu_{1}(X_j)-\hat{\mu}_{1}(X_j) \right) \\
    & - \sum_j (1-T_j) \omega_{ij0} \left( \hat{\mu}_{0}(X_i)-\mu_{0}(X_i)+\mu_{0}(X_j)-\hat{\mu}_{0}(X_j) \right)\Bigg) \\
    & = \meanin 
    \Bigg( \sumjn T_j \omega_{ijT_j} \left(\hat{\mu}_{T_j}(X_i)-\mu_{T_j}(X_i)+\mu_{T_j}(X_j)-\hat{\mu}_{T_j}(X_j) \right) \\
    & - \sum_j (1-T_j) \omega_{ijT_j} \left( \hat{\mu}_{T_j}(X_i)-\mu_{T_j}(X_i)+\mu_{T_j}(X_j)-\hat{\mu}_{T_j}(X_j) \right)\Bigg) \\
    & = \meanin \sumjn (2 T_j - 1) \omega_{ijT_j} \left(\hat{\mu}_{T_j}(X_i)-\mu_{T_j}(X_i)- (\hat{\mu}_{T_j}(X_j)-\mu_{T_j}(X_j)) \right)
\end{align*}

The first term, $S$, is asymptotically normal while $\tilde{R}$ is $\sqrt{n}$-negligible.

\paragraph{Main term} 

Consider $S$ first: Let $\theta_i = \sumjn \omega_{ji T_i}$ and decompose
\begin{align*}
S = \frac{1}{n} \sum_{i=1}^n (2T_i - 1) \theta_i \varepsilon_i 
&= \frac{1}{n} \sum_{i=1}^n (2T_i - 1) \left( \frac{T_i}{p(X_i)} + \frac{1 - T_i}{1 - p(X_i)} \right) \varepsilon_i \\
& + \frac{1}{n} \sum_{i=1}^n (2T_i - 1) \left(\theta_i - \frac{T_i}{p(X_i)} - \frac{1 - T_i}{1 - p(X_i)} \right) \varepsilon_i.
\end{align*}

Then,
\begin{equation*}
\mathds{E} \left[ \frac{1}{n} \sum_{i=1}^n (2T_i - 1) \left(\theta_i-\frac{T_i}{p(X_i)} - \frac{1 - T_i}{1 - p(X_i)}\right) \epsilon_i \, \middle| \, \boldsymbol{X}, \boldsymbol{T} \right] = 0,
\end{equation*}
\begin{align*}
\mathds{V}\left[\frac{1}{n} \sumin (2T_i - 1) \left(\theta_i-\frac{T_i}{p(X_i)} - \frac{1 - T_i}{1 - p(X_i)}\right) \epsilon_i \, \middle\vert \, \boldsymbol{X}, \boldsymbol{T} \right] = \frac{1}{n^2} \sum_{i=1}^n \left(\theta_i-\frac{T_i}{p(X_i)} - \frac{1 - T_i}{1 - p(X_i)}\right)^2 \sigma^2_{T_i}(X_i),
\end{align*}
and
\begin{equation*}
\frac{1}{n} \sum_{i=1}^n (2T_i - 1) \left(\theta_i-\frac{T_i}{p(X_i)} - \frac{1 - T_i}{1 - p(X_i)}\right) \epsilon_i = o_P(n^{-1/2}).
\end{equation*}
Then,
\begin{align*}
\meanin (2 T_i - 1) \theta_i \varepsilon_i + \meanin \mu_1(X_i) - \mu_0(X_i) - \mbox{ATE} &= \frac{1}{n} \sum_{i=1}^n \left[ \mu_1(X_i) - \mu_0(X_i) - \mbox{ATE} \right] \\
& + \frac{1}{n} \sum_{i=1}^n (2T_i - 1) \left( \frac{T_i}{p(X_i)} + \frac{1 - T_i}{1 - p(X_i)} \right) \epsilon_i \\
& + o_P(n^{-1/2}).
\end{align*}
Hence, by the central limit theorem,
\begin{align}
\begin{split}
    & \sqrt{n} \left(\mathds{E}\left[\frac{\mathds{V}[Y_i(1)\vert X_i]}{p(X_i)}+\frac{\mathds{V}[Y_i(0)\vert X_i]}{1-p(X_i)}\right] +
    \mathds{V}[\mbox{CATE}(X_i)] \right)^{-1/2} \\
    \times & \left( \meanin (2 T_i - 1) \theta_i \varepsilon_i + \meanin \mu_1(X_i) - \mu_0(X_i) - \mbox{ATE} \right) \xrightarrow{d} \mathcal{N}(0, 1)
\end{split}
\end{align}

\paragraph{Remainder}

The remaining term vanishes at a faster rate. \\
Let $\hat{\mu}_{t}^*(X_i) \equaldef \frac{1}{\sumkn \mathds{1}(T_k=t) \mathds{P}[k \in \mathcal{C}_{i} \vert X_i, X_k]} \sumjn (\mathds{1}(T_j=t) \mathds{P}[j \in \mathcal{C}_{i}\vert X_i, X_j] Y_j)$ and note that $\mathds{P}[j \in \mathcal{C}_{it}\vert X_i, X_j]$ inherits smoothness properties from $g$, in particular, it is differentiable $K$ times with uniformly bounded $K$-th derivatives.\footnote{When $M$ grows, the bound scales with $M=o(\ln(n))$, which does not affect the negligibility of the relevant terms.} \medskip

Noting that 
$\hat{\mu}_{T_j}(X_i)-\hat{\mu}_{T_j}(X_j) - (\mu_{T_j}(X_i)-\mu_{T_j}(X_j)) = \hat{\mu}_{T_j}(X_i)-\hat{\mu}_{T_j}(X_j) - (\hat{\mu}_{T_j}^*(X_i)-\hat{\mu}_{T_j}^*(X_j)) + 
(\hat{\mu}_{T_j}^*(X_i)-\hat{\mu}_{T_j}^*(X_j)) - (\mu_{T_j}(X_i)-\mu_{T_j}(X_j))$, the remainder $\tilde{R}$ can be decomposed into $\tilde{R}_1$ and $\tilde{R}_2$. Consider first 
\begin{equation}
\tilde{R}_1 =
\meanin \sumjn (2 T_j - 1) \omega_{ijT_j} \left(\hat{\mu}_{T_j}^*(X_i)-\mu_{T_j}(X_i)- (\hat{\mu}_{T_j}^*(X_j) - \mu_{T_j}(X_j)) \right) 
\end{equation}
We have
\begin{equation}
    \lvert \tilde{R}_1 \rvert \leq
\meanin \sumjn \max_{t \in\{0,1\}} \omega_{ij T_j}
\left\rvert \mu_t(X_i) - \mu_t(X_j) - \hat{\mu}_t^*(X_i) + \hat{\mu}_t^*(X_j)
\right\rvert.
\end{equation}
\noindent For $t \in \{0,1\}$, using a Taylor expansion to the $L$-th order,
\begin{equation}
\left| \mu_t(X_j) - \mu_t(X_i) - \sum_{\ell=1}^{L-1}
\frac{1}{\ell!} \sum_{\alpha \in\Lambda_\ell} \partial^\alpha \mu_t(X_i)(X_j-X_i)^\alpha \right|
\leq \max_{\alpha \in\Lambda_L} \|\partial^\alpha\mu_t\|_\infty \frac{1}{L!} \sum_{\alpha\in\Lambda_L} \|X_j-X_i\|^L,
\end{equation}
where $\Lambda_\ell$ is the set of natural numbers $(\alpha_1, \ldots, \alpha_d)$ that sum up to $\ell$.

\noindent Analogously, 
\begin{equation}
\left|
\hat{\mu}_t^*(X_j) - \hat{\mu}_t^*(X_i)
-\sum_{\ell=1}^{L-1} \frac{1}{\ell!} \sum_{\alpha \in\Lambda_\ell}
\partial^\alpha \hat{\mu}_t^* (X_i)(X_j-X_i)^\alpha \right|
\leq
\max_{\alpha \in\Lambda_L}
\|\partial^\alpha\hat{\mu}_t^*\|_\infty
\frac{1}{L!}\sum_{\alpha\in\Lambda_L} \|X_j-X_i\|^L.
\end{equation}

\noindent Matching derivatives, we get
\begin{equation}
\left|
\sum_{\ell=1}^{L-1} \frac{1}{\ell!} \sum_{\alpha \in\Lambda_\ell}
\left(\partial^\alpha\hat{\mu}_t^*(X_i) - \partial^\alpha\mu_t(X_i) \right) (X_j-X_i)^\alpha
\right|
\leq \sum_{\ell=1}^{L-1} \max_{\alpha \in \Lambda_\ell}
\|\partial^\alpha\hat{\mu}_t^*-\partial^\alpha \mu_t\|_\infty
\frac{1}{\ell!} \sum_{\alpha \in\Lambda_\ell} \lVert X_j-X_i\rVert^\ell .
\end{equation}
As a result, 
\begin{equation}
\begin{split}
\tilde{R}_1 &\leq C  \left(
\max_{t\in\{0,1\}}
\max_{\alpha\in\Lambda_L}
\|\partial^\alpha \mu_t\|_\infty
+
\max_{t\in\{0,1\}} \max_{\alpha\in\Lambda_L}
\|\partial^\alpha\hat{\mu}_t^*\|_\infty
\right)
\left(
\meanin \sumjn \omega_{ij T_j} \lVert X_j-X_i\rVert^L
\right) \\
& + \sum_{\ell=1}^{L-1}
\left(
\max_{t\in\{0,1\}} \max_{\alpha\in\Lambda_\ell}
\|\partial^\alpha\hat{\mu}_t^*-\partial^\alpha\mu_t\|_\infty
\right)
\left( \meanin \sumjn \omega_{ij T_j} \|X_j-X_i\|^\ell \right).
\end{split}
\end{equation}

Proceeding as in the Proof of Theorem 4.1(iii) in \citet{lin2025regression}'s online appendix (with straightforward adaptations as in the proof of Lemma \ref{weights_to_score_lemma}), it follows that 
\begin{equation}
    \mathds{E}\left[\sum_j \omega_{ij T_j} \lVert X_j - X_i\rVert^\ell\right] \leq C \int g(\lVert y\rVert) \lVert y\rVert^\ell \mbox{d}y {s_n}^{-\ell};
\end{equation}
which is $0$ for $\ell = 1$ and $O({s_n}^{-\ell})$ for $\ell = 2, \ldots, L$. Then, noting that the variance dominates the bias, standard results on kernel estimation yields $\max_{\alpha, t} \lVert\partial^\alpha \hat{\mu}_t^* - \partial^\alpha \mu_t||_{\infty} = O_p((\psi_i \lambda_n)^{-M} {s_n}^{2\ell} \ln(n)$, the condition that $n^{1/2} (\psi_i \lambda_n)^{-M} {s_n}^{\ell} \ln(n) \rightarrow 0$ implies $n^{1/2}\tilde{R}_2 \rightarrow 0$. 

\noindent Now, noting that $\sumjn T_j \omega_{ij1} = 1$ and using Lemma \ref{weights_to_score_lemma}, the second piece is given by
\begin{align}
    \begin{split}
        \tilde{R}_2 &=
        \meanin \sumjn (2 T_j - 1) \omega_{ij T_j} \left(
        \hat{\mu}_{T_j}(X_i) - \hat{\mu}_{T_j}(X_j) - \hat{\mu}_{T_j}^*(X_i) + \hat{\mu}_{T_j}^*(X_j)
        \right) \\
        & = \meanin \sumjn \Bigg( T_j \omega_{ij1} (\hat{\mu}_{1}(X_i) - \hat{\mu}_{1}(X_j) - \hat{\mu}_{1}^*(X_i) \\
        & + \hat{\mu}_{1}^*(X_j)) - (1-T_j) \omega_{ij0} (\hat{\mu}_{0}(X_i) - \hat{\mu}_{0}(X_j) - \hat{\mu}_{0}^*(X_i) + \hat{\mu}_{0}^*(X_j)) \Bigg) \\
        & = \meanin \sumjn T_j \omega_{ij1} (\hat{\mu}_1(X_i) - \hat{\mu}_1^*(X_i)) + \meanin \sumjn T_j \omega_{ij1} (\hat{\mu}_1^*(X_j) - \hat{\mu}_1(X_j)) \\
        & - \meanin \sumjn (1-T_j) \omega_{ij0} (\hat{\mu}_0(X_i) - \hat{\mu}_0^*(X_i)) - \meanin \sumjn (1-T_j) \omega_{ij0} (\hat{\mu}_0^*(X_j) - \hat{\mu}_0(X_j)) \\
        & = \meanin (\hat{\mu}_1(X_i) - \hat{\mu}_1^*(X_i)) - \meanin \sumjn T_i \omega_{ji1} (- \hat{\mu}_1^*(X_i) + \hat{\mu}_1(X_i)) \\
        & - \meanin (\hat{\mu}_0(X_i) - \hat{\mu}_0^*(X_i)) + \meanin \sumjn (1-T_i) \omega_{ji0} (-\hat{\mu}_0^*(X_i) + \hat{\mu}_0(X_i)) \\
        & = \meanin \left(1-T_i \sumjn \omega_{ji1}\right) (\hat{\mu}_1(X_i) - \hat{\mu}_1^*(X_i)) \\
        & - \meanin \left(1-(1-T_i) \sumjn \omega_{ji0}\right) (\hat{\mu}_0(X_i) - \hat{\mu}_0^*(X_i)) \\
        & = \meanin \left(1-\frac{T_i}{p(X_i)} \right) (\hat{\mu}_1(X_i) - \hat{\mu}_1^*(X_i)) - \meanin \left(1-\frac{1-T_i}{1-p(X_i)}\right) (\hat{\mu}_0(X_i) - \hat{\mu}_0^*(X_i)) + \tilde{r}_2,
    \end{split}
\end{align}
where $\tilde{r}_2$ is of lower order.

Consider the first term, $\tilde{R}_{21} \equaldef \meanin \left(1-\frac{T_i}{p(X_i)} \right) (\hat{\mu}_1(X_i) - \hat{\mu}_1^*(X_i))$. The second can be treated analogously. $\tilde{R}_{21}$ can be further decomposed into $\tilde{R}_{211}+\tilde{R}_{212}$, where 

\begin{align}
    \tilde{R}_{211} & = \meanin \left(1-\frac{T_i}{p(X_i)} \right) (\hat{\mu}_1(X_i) - \tilde{\mu}_1(X_i))\\
    \tilde{R}_{212} & = \meanin \left(1-\frac{T_i}{p(X_i)} \right) (\tilde{\mu}_1(X_i) - \hat{\mu}_1^*(X_i))
\end{align}
where $\tilde{\mu}_1(X_i) \equaldef \frac{1}{\sumkn T_k \mathds{P}[k \in \mathcal{C}_{i} \vert X_i, X_k]} \sumjn \mathds{1}(j \in \mathcal{C}_{i1}) Y_j$.

The two terms can be handled similarly, after linearizing in the case of $\tilde{R}_{211}$. Consider $\tilde{R}_{212}$. We have 

\begin{equation}
    \mathds{E}[\tilde{R}_{212}] = \mathds{E}[\mathds{E}[\tilde{R}_{212} \vert X]] = 0.
\end{equation}

Noting that, for $i\neq j$,

\begin{align*}
    &\mathds{C}ov[(1-T_i/p(X_i)) (\tilde{\mu}_1(X_i) - \hat{\mu}_1^*(X_i)); (1-T_j/p(X_j)) (\tilde{\mu}_1(X_j) - \hat{\mu}_1^*(X_j)) \vert \boldsymbol{X}] \\
    &= \mathds{E}[(1-T_i/p(X_i)) (\tilde{\mu}_1(X_i) - \hat{\mu}_1^*(X_i)) (1-T_j/p(X_j)) (\tilde{\mu}_1(X_j) - \hat{\mu}_1^*(X_j)) \vert \boldsymbol{X}] \\
    & = \mathds{E}[\mathds{E}[(1-T_i/p(X_i)) (\tilde{\mu}_1(X_i) - \hat{\mu}_1^*(X_i)) (1-T_j/p(X_j)) (\tilde{\mu}_1(X_j) - \hat{\mu}_1^*(X_j)) \vert \boldsymbol{X}, \boldsymbol{Y}(1), \boldsymbol{T}]\vert \boldsymbol{X}] \\
    & = 0
\end{align*}
since $\mathds{E}[(\tilde{\mu}_1(X_j) - \hat{\mu}_1^*(X_j)) \vert \boldsymbol{X}, \boldsymbol{Y}(1), \boldsymbol{T}] = 0$, it follows that the variance obeys
\begin{align}
    \begin{split}
        \mathds{V}[\tilde{R}_{212}] & = \mathds{E}[\mathds{V}[\tilde{R}_{212} \vert X]] + \mathds{V}[\mathds{E}[\tilde{R}_{212}\vert \boldsymbol{X}]] \\
        & = n^{-2} \sumin \mathds{E}[\mathds{V}[(1-T_i/p(X_i)) (\tilde{\mu}_1(X_i) - \hat{\mu}_1^*(X_i))\vert \boldsymbol{X}]] \\
        & = n^{-2} \sumin \mathds{E}[\mathds{E}[(1-T_i/p(X_i))^2 (\tilde{\mu}_1(X_i) - \hat{\mu}_1^*(X_i))^2 \vert \boldsymbol{X}]] \\ 
        & = n^{-2} \sumin \mathds{E}[\mathds{E}[(1-T_i/p(X_i))^2\vert X] \mathds{E}[(\tilde{\mu}_1(X_i) - \hat{\mu}_1^*(X_i))^2\vert \boldsymbol{X}]] \\
        & = n^{-2} \sumin \mathds{E}[(1-p(X_i))p(X_i)^{-1} \mathds{E}[(\tilde{\mu}_1(X_i) - \hat{\mu}_1^*(X_i))^2\vert \boldsymbol{X}]] \\ 
        & \leq C/n \mathds{E}[(\tilde{\mu}_1(X_i) - \hat{\mu}_1^*(X_i))^2] \\
        & = o(1/n).
    \end{split}
\end{align}

As a result, $\sqrt{n} \tilde{R}_2 \overset{p}{\longrightarrow} 0$.

Finally, the proposed estimator for the asymptotic variance is an adapted version of the one proposed by \citet{lin2025regression} and its consistency can be established similarly. 



\section*{Appendix B: Further results}

\subsection{Clustering}

\noindent Since the covariate density is continuous (then bounded by compactness), the clustering coefficient is
\begin{align*}
    C & \equaldef \mathds{P}[W_{jk}=1\vert W_{ij}=W_{ik}=1] \\
    & = \mathds{E}[\mathds{P}[W_{jk}=1\vert W_{ij}=W_{ik}=1, X_j, X_k]\vert W_{ij}=W_{ik}=1] \\
    & = \mathds{E}[\mathds{P}[W_{jk}=1\vert X_j, X_k]\vert W_{ij}=W_{ik}=1] \\
    & = \mathds{E}[g(s_n \Vert x_j - x_k \Vert)\vert W_{ij}=W_{ik}=1] \\
    & = \int_{\mathds{R}^{d}} \int_{\mathds{R}^{d}} g(s_n \Vert x_j - x_k \Vert) f_{X_j, X_k\vert W_{ij}=1, W_{ik}=1}(x_j, x_k) \ \mbox{d}x_j \mbox{d}x_k\\ 
    & = \int_{\mathds{R}^{d}} \int_{\mathds{R}^{d}} g(s_n \Vert x_j - x_k \Vert) \frac{\mathds{P}[{W_{ij}=1=W_{ik}\vert X_j=x_j, X_k=x_k}] f(x_k) f(x_j)}{\mathds{P}[W_{ij}=1=W_{ik}]} \ \mbox{d}x_j \mbox{d}x_k \\
    & = \int_{\mathds{R}^{d}} \int_{\mathds{R}^{d}} g(s_n \Vert x_j - x_k \Vert) \frac{\mathds{E}[\mathds{P}[W_{ij}=1=W_{ik}\vert X_i=x_i, X_j=x_j, X_k=x_k]\vert X_j=x_j, X_k=x_k]}{\mathds{E}[\mathds{P}[W_{ij}=1=W_{ik}\vert X_i=x_i, X_j=x_j, X_k=x_k]]} \\ 
    & f(x_k) f(x_j) \ \mbox{d}x_j \mbox{d}x_k \\
    & = \int_{\mathds{R}^{d}} \int_{\mathds{R}^{d}} \frac{g(s_n \Vert x_j - x_k \Vert) \int_{\mathds{R}^{d}} g(s_n \Vert x_i - x_j \Vert)  g(s_n \Vert x_i - x_k \Vert) f(x_i) \ \mbox{d}x_i}{\int_{\mathds{R}^{3d}} g(s_n \Vert x_i - x_k \Vert)  g(s_n \Vert x_i - x_j \Vert) f(x_i) f(x_j) f(x_k) \ \mbox{d}x_i \mbox{d}x_j \mbox{d}x_k}  \ f(x_j) f(x_k) \ \mbox{d}x_j \mbox{d}x_k \\
    & = \frac{{s_n}^{-3d}}{{s_n}^{-3d}} \int_{\mathds{R}^{3d}} \frac{g(\Vert \hat{x}_j - \hat{x}_k \Vert) g(\Vert \hat{x}_i - \hat{x}_j \Vert) g(\Vert \hat{x}_i - \hat{x}_k \Vert) f\left(\frac{\hat{x}_i}{s_n}\right) f\left(\frac{\hat{x}_j}{s_n}\right) f\left(\frac{\hat{x}_k}{s_n}\right)}{\int_{\mathds{R}^{3d}} g(\Vert \hat{x}_i - \hat{x}_k \Vert) g(\Vert \hat{x}_i - \hat{x}_j \Vert) f\left(\frac{\hat{x}_i}{s_n}\right) f\left(\frac{\hat{x}_j}{s_n}\right) f\left(\frac{\hat{x}_k}{s_n}\right) \ d\hat{x}_i d\hat{x}_j d\hat{x}_k} \ d\hat{x}_i d\hat{x}_j d\hat{x}_k \\
    & \rightarrow \frac{\int_{\mathds{R}^{3d}} g(\Vert \hat{x}_j - \hat{x}_k \Vert) g(\Vert \hat{x}_i - \hat{x}_j \Vert) g(\Vert \hat{x}_i - \hat{x}_k \Vert) d\hat{x}_i d\hat{x}_j d\hat{x}_k}{\int_{\mathds{R}^{3d}} g(\Vert \hat{x}_i - \hat{x}_j \Vert) g(\Vert \hat{x}_i - \hat{x}_k \Vert) \ d\hat{x}_i d\hat{x}_j d\hat{x}_k}
\end{align*}
where a hat indicates a change of variable of the form $\hat{x} = s_n (x-0)$, assuming for simplicity that $0$ belongs to the support of $x$.

\subsection{Dense networks with homophily in the unobservables}\label{section_hiu}

This section considers a link formation model that only assumes homophilic behavior in unobservables and is suitable for dense networks. The method provides insight about how to create sub-groups that are increasingly close in terms of unobservables and can be adapted to deal with empirical concerns such as matching on treatment status. Although dense networks are less frequent, this approach is thus valuable as it covers additional network structures and applications of interest. \medskip

Now, a link exists between $i$ and $j$ if $\eta_{ij} \leq w(h(X_i^o; X_j^o) + \Vert X_i^u - X_j^u\Vert)$. The function $w$ may not depend on $n$ anymore. The function $h$ need not be homophilic nor separable in the observed $X^o$ but is assumed known. Most results would apply with minor modifications if a lower bound with the relevant properties can be obtained. The dimension of $\mathcal{X}^u$ is denoted by $d_u$. \medskip

When observables affect the outcome, they can be controlled for using standard methods.\footnote{For instance, by constructing again products with kernel weights or performing regression adjustments.} For ease of exposition, I focus on presenting how to control for the unobserved components. \medskip

I consider again $\frac{1}{\vert \mathcal{C}_{i1}(\kappa)\vert} \sum_{j \in \mathcal{C}_{i1}(\kappa)} Y_j - \frac{1}{\vert \mathcal{C}_{i0}(\kappa)\vert} \sum_{j \in \mathcal{C}_{i0}(\kappa)} Y_j$, where the comparison group now depends on a truncation parameter $\kappa$. It will play a key role by placing a lower bound on $h$, inducing a closer distribution of unobservables among friends. \medskip

The main idea is that if there is no observed rationale for two people being friends, it becomes more likely that there is an unobserved reason for their friendship. Then, people who are friends despite a high value of $h$ are less likely to differ strongly on unobservables. To see this, consider the simplified case where people reject friendship with anyone whose quality of match does not meet a certain threshold (\textit{i.e.}, $w(x)=0$ for any $x$ large enough). Then, two friends with an $h$ close to the boundary must have close unobservables since a high discrepancy in unobservables would have brought $h+\Vert X_i^u-X_j^u\Vert$ above the threshold. \medskip

This suggests using comparison groups of the form $\mathcal{C}_i(\kappa) \equaldef \{j \in \mathcal{N}(i)\vert h(x_i^o, x_j^o) > \kappa\}$. The estimator then truncates the sums to select individuals whose observed characteristics make them unlikely to be friends. $\kappa$ is viewed as a sequence converging to $\infty$ at a rate to be determined. Using this strategy, a counterpart to the main theorems can be established with $\kappa$-truncation replacing asymptotic homophily.

\begin{theorem}\label{CATE_unobs_CAN}
Suppose $\CATE(x)$ is Hölder continuous with exponent $\alpha$ on a neighborhood of $x_i$, $w$ has bounded support, and there exist a sequence $\lambda_n \rightarrow \infty$ and a sequence $b_n$ such that $\kappa$-truncation satisfies $n w(\kappa+b_n) b_n^{d_u+1} \geq C \lambda_n$ and a sequence $\varepsilon_n \downarrow 0$ satisfying $\kappa+\varepsilon_n > \sup\{\mbox{supp}\{w\}\}$ eventually. Then, the estimator satisfies
\begin{equation*}
    \sqrt{\vert \mathcal{C}_i\vert} (\widehat{\CATE}(x_i; \mathcal{C}_i(\kappa))-\CATE(x_i)-\mathds{B}_i) \overset{d}{\rightarrow} \mathcal{N}\left(0; V\right)
\end{equation*}
and the bias is negligible if $\sqrt{\lambda_n} {\varepsilon_n}^\alpha \rightarrow 0$.
\end{theorem}

\begin{proof}

To ease notation, the argument $\kappa$ is omitted in all instances of $\mathcal{C}_{it}(\kappa)$. I also make use of the following shorthands: $\Delta_{ij}^u \equaldef \Vert x_j^u -x_i^u\Vert$, $h_{ij} \equaldef h(X_j^o, X_i^o)$. Decompose the centered mean of the group as
\begin{equation*}
    \frac{1}{\vert \mathcal{C}_{it}\vert} \sum_{j \in \mathcal{C}_{it}} Y_j(T_j) - \mathds{E}[Y(t)\vert x_i] = \frac{1}{\vert \mathcal{C}_{it} \vert} \sum_{j \in \mathcal{C}_{it}} Y_j(T_j) - \mathds{E}[Y_j(t)\vert \mathcal{C}_{it}] + \mathds{E}[Y_j(t)\vert \mathcal{C}_{it}] - \mathds{E}[Y(t)\vert x_i]
\end{equation*}
for $t \in \{0, 1\}$.

The first term depends on sample fluctuations and converges (in probability to $0$ and in distribution once re-scaled) provided the number of observations in the sum grows to infinity. The second term is a bias term and disappears under regularity conditions and the truncation $h_{ij} > \kappa$ with $\kappa \rightarrow \infty$. Thus, the main steps prove that (i) the number of observations in the sum grows to infinity and (ii) the bias disappears, for some sequence $\kappa \rightarrow \infty$. Consider the probability of an observation belonging to $\mathcal{C}_{i}$ first: 
\begin{align*}
    \begin{split}
    \mathds{P}[\mathcal{C}_{i}] & \geq C \int_{h_{ij} > \kappa} w_{ij} f(x_j) \ \mbox{d}x_j \\
    & \geq C \int_{\Delta_{ij}^u \leq \frac{b_n}{2}, \kappa < h_{ij} < \kappa+\frac{b_n}{2}} w(\kappa+b_n) f(x_j) \ \mbox{d}x_j \\
    & = C w(\kappa+b_n) \mathds{P}\left[\Delta_{ij}^u \leq \frac{b_n}{2}, \kappa < h_{ij} < \kappa+\frac{b_n}{2}\right] \\
    & = C w(\kappa+b_n) \mathds{P}\left[\Delta_{ij}^u \leq \frac{b_n}{2}\right] \mathds{P}\left[\kappa < h_{ij} < \kappa+\frac{b_n}{2} \middle\vert \Delta_{ij}^u \leq \frac{b_n}{2}\right] \\
    & \geq C w(\kappa+b_n) {b_n}^{d_u+1}
    \end{split}
\end{align*}

Following previous proofs, it suffices to let $\kappa \rightarrow \infty$ slowly enough to induce a rate of $\frac{\lambda_n}{n}$ for the above probability. Next, the bias term reads
\begin{align*}
\begin{split}
    \left\vert \mathds{E}[Y_j(T_j)\vert \mathcal{C}_{it}] - \mathds{E}[Y_j(t)\vert x_i]\right\vert & = \left\vert  \mathds{E}[\mathds{E}[Y_j(T_j)\vert X_j]\vert \mathcal{C}_{it}] - \mathds{E}[Y_j(t)\vert x_i]\right\vert  \\
    & = \left\vert \int_{\mathds{R}^d} (\mathds{E}[Y_j(t)\vert x_j] - \mathds{E}[Y_j(t)\vert x_i]) f_{X_j\vert \mathcal{C}_{it}}(x_j) \ \mbox{d}x_j \right\vert \\
    & \leq \left\vert \int_{B_\varepsilon(x_i)} (\mathds{E}[Y_j(t)\vert x_j] - \mathds{E}[Y_j(t)\vert x_i]) f_{X_j\vert \mathcal{C}_{it}}(x_j) \ \mbox{d}x_j \right\vert \\
    & + \left\vert \int_{B_\varepsilon^c(x_i)} (\mathds{E}[Y_j(t)\vert x_j] - \mathds{E}[Y_j(t)\vert x_i]) f_{X_j\vert \mathcal{C}_{it}}(x_j) \ \mbox{d}x_j \right\vert \\
    & \leq C \varepsilon^\alpha \\
    & + \frac{1}{\mathds{P}[\mathcal{C}_{it}]} \int_{B_\varepsilon^c(x_i)} (\mathds{E}[Y_j(t)\vert x_j] - \mathds{E}[Y_j(t)\vert x_i]) \mathds{P}[\mathcal{C}_{it}\vert x_j] f_{X_j, T_j}(x_j, t) \ \mbox{d}x_j \\
    & \leq C \varepsilon^\alpha \\
    & + \frac{C}{\mathds{P}[\mathcal{C}_{it}]} \int_{B_\varepsilon^c(x_i), h_{ij} > \kappa} (\mathds{E}[Y_j(t)\vert x_j] - \mathds{E}[Y_j(t)\vert x_i]) w_{ij} f_{X_j, T_j}(x_j, t) \ \mbox{d}x_j \\
    & \leq C \varepsilon^\alpha + \frac{C n}{\lambda_n} 
    w(\kappa+ \varepsilon) \left(\mathds{E}[\vert \mathds{E}[Y_j(t)\vert X_j]\vert \vert T_j=t] + \vert \mathds{E}[Y_j(t)\vert x_i]\right)
\end{split}
\end{align*}

\noindent so that, with $\varepsilon \downarrow 0$, the bias disappears as $\kappa$ rises provided $w(\kappa + \varepsilon) = o(\frac{\lambda_n}{n})$.

\end{proof}

The condition on $w(\kappa+b_n) b_n^{d_u+1}$ restricts the speed at which $\kappa$ can increase so that the number of observations used in estimating the CATE keeps growing. The first part pertains to the behavior of the $w$ function; the second term pertains to the space in which unobservables live. \medskip

The term $w(\kappa+b_n)$ comes from the increasing cost of truncating as potential connections are accepted at decreasing rates. In the presence of a discontinuity at the end of the support, \textit{i.e.}, $w(x) = a \mathds{1}_{x \leq D}$ for $a \in ]0; 1]$ and $D\in \mathds{R}^+$, this term disappears. The second term, ${b_n}^{d_u+1}$, is the result of forcing unobservables in a $b_n$-ball using values of $h$ lying between $\kappa$ and $\kappa+b_n$. \medskip

A natural estimator of the end of support, if unknown, is the highest value of $h$ among all $i, j$ satisfying $W_{ij}=1$. Unbounded support of $w$ can be accommodated, provided the function vanishes sufficiently quickly (the threshold is exponential: functions decaying faster than $w=e^{-x}$ provide a sufficiently fast decay). In this case, a factor of $\rho_h(\kappa + b_n/2)$, where $\rho_h$ is (a lower bound on) the tail decay of $h$ conditional on $X_j^u \in B_{r}(x_i^u)$ for some $r$, has to be added. The reason is the tail decay of the density while one seeks increasingly larger values of $h$. \medskip

The conditions on $\varepsilon_n$ ensure that the bias disappears sufficiently fast to enable standard inference. A more primitive statement is $w(\kappa + \varepsilon_n) = o(\lambda_n/n)$, which mirrors conditions such as those in Assumption \ref{Hölder}, part b); this boils down to $\kappa + \varepsilon_n$ eventually crossing the end of the support of $w$ when it is finite. \medskip

Finally, alternative comparison groups can again be considered, by including friends of friends or people with sufficiently many common friends, or by constructing the truncated group differently. For instance, considering triangles of friends can give more leeway to vary $h$. \medskip 

Overall, these results show that suitably refining a comparison group such as friends using the observed covariates allows one to isolate increasingly good matches in terms of unobservables. Although the levels of the unobserved variables are not identified, groups with increasingly similar values can be recovered. The rates are slower than those arising from asymptotically homophilic networks, though it should be noted that they focus directly on the unobserved components while the observed variables can be adjusted in more traditional ways, \textit{e.g.}, regression adjustments. This is sufficient to recover consistent estimators of treatment effects under unobserved confounders with dense networks and to learn who is comparable to whom in terms of unobservables. 


\subsection{Peer effects}



This subsection discusses the consequences of relaxing SUTVA in a simple linear-in-means model. Let \(\overline{W}\) be the row-normalized adjacency matrix and assume regular asymptotic homophily.\footnote{To simplify the exposition, this section assumes that the weighting matrix that describes the relevant peer groups (the matrix relevant for peer effects) coincides with $W$ (the matrix used as information about unobserved confounders) and that the whole vector $X$ from the unconfoundedness condition informs network formation in an asymptotically homophilic way. This is not crucial: (endogenous) peer effects induce a dependence on neighbor's covariates, which can be controlled for through $X^o$ and possibly another network, and on some of neighbor's idiosyncratic components, which affects inference but not the identification strategy.}\footnote{Asymptotic homophily in this context generates reference groups similar to those in \citet{manski1993identification} while allowing for imperfect matching.} Consider the standard linear-in-means outcome equation in matrix form, where the peer effect parameter $\rho \in ]-1;1[$:
\begin{equation}
    Y=\rho \overline{W}Y+X\beta+\mbox{CATE}(X) T+\varepsilon,
\end{equation}
so that
\begin{equation}
    Y=S(X\beta+\mbox{CATE}(X) T+\varepsilon),
\end{equation}
where $S\equaldef (I-\rho \overline{W})^{-1}$.

Let $\mathds{E}[\varepsilon_i\mid X_i,T_i,W]=0$, with bounded $2+\delta$ moments. Notice that for any smooth \(\varphi\), $\overline{W}^k\varphi(X) - \varphi(X_i)\rightarrow 0$ for fixed \(k\), while for any mean-zero idiosyncratic component $u_i$ satisfying $\mathds{E}[u_i\mid X_i]=0$, we have $(\overline{W}^k u)_i=O_p(1)$ ($o_p(1)$ whenever $\lambda=\infty$) by regular asymptotic homophily.

Since 
\begin{align*}
    &S (X_i'\beta+\mbox{CATE}(X_i) p(X_i))\\
    &= \sum_{k=0}^{\infty}\rho^k \overline{W}^k (X_i'\beta+\mbox{CATE}(X_i) p(X_i))\\
    &= \frac{X_i'\beta+\mbox{CATE}(X_i) p(X_i)}{1-\rho}+o_p(1).
\end{align*}
and
\begin{align*}
    &S(\mbox{CATE}(X_i)(T_i-p(X_i))+\varepsilon_i)\\
    &= \mbox{CATE}(X_i)(T_i-p(X_i))+\varepsilon_i+\sum_{k=1}^{\infty}\rho^k \overline{W}^k \mbox{CATE}(X_i)(T_i-p(X_i))+\varepsilon_i\\
    &= \mbox{CATE}(X_i)(T_i-p(X_i))+\tilde{\varepsilon}
    _i+o_p(1),
\end{align*}
where $\tilde{\varepsilon}_i$ contains mixture of idiosyncratic errors among neighbors and $\tilde{\varepsilon}_i=\varepsilon_i$ if $\lambda = \infty$, the equilibrium outcome admits the expansion
\begin{equation}
    Y_i = \frac{X_i'\beta+\mbox{CATE}(X_i) p(X_i)}{1-\rho}
    + \mbox{CATE}(X_i)(T_i-p(X_i))
    + \tilde{\varepsilon}_i + o_p(1).
\end{equation}

For linked units $i, j$, regular asymptotic homophily implies\footnote{If $\overline{W}$ is not purely asymptotically homophilic in the whole $X$, then the difference depends on some mixture of covariates for which adjustments based on observables or another network must be performed to retain consistency.}
\begin{equation*}
    Y_j-Y_i = \mbox{CATE}(X_i)(T_j-T_i)+\tilde{\varepsilon}_j-\tilde{\varepsilon}_i+o_p(1).
\end{equation*}

\begin{proposition}[Direct effects under linear peer effects]
Under the framework discussed above, the CATE estimators satisfy $\widehat{\CATE}(X_i) \to_p \mbox{CATE}(X_i)$ if the regularity conditions of Theorem \ref{CATE_consistency} are satisfied.
\end{proposition}

Peer effects do not affect consistency of the CATE estimators; they change the interpretation of the estimand, insofar as they capture direct effects but neglect social multiplier effects. \medskip 

Inference should account for a potential equilibrium correlation in the outcomes. If
$\xi\equiv S\varepsilon$, then, conditional on $W$,
\[
    \mathds{V}(\xi\mid W)=\sigma_\varepsilon^2SS',
\]
under homoskedastic independent structural shocks. The CATE estimator for individual $i$ can be written as $\omega_i' Y$, so that the conditional variance reads
\[
    \mathrm{Var}(\widehat{\CATE}(X_i)\mid W,X,T)
    = \omega_i'\Omega \omega_i,
    \qquad
    \Omega=\mathds{V}[\xi\mid W].
\]
A nonparametric way to estimate this variance is a network HAC estimator such as
\begin{equation*}
    \widehat V_i = \sum_{j,\ell}\omega_{ij}\omega_{i\ell}\hat\xi_j\hat\xi_\ell
    K\!\left(\frac{\mathrm{dist}_W(j,\ell)}{L_n}\right),
\end{equation*}
where \(\mathrm{dist}_W(j,\ell)\) is graph distance, \(K\) is a bounded kernel, and \(L_n\) is a graph-distance bandwidth. A simpler sensitivity analysis fixes a grid \(\mathcal R\) of plausible values of $\rho$ and reports
\[
    \widehat V_i(\rho)
    =
    \omega_i'\widehat\Omega(\rho)\omega_i,
    \qquad
    \widehat\Omega(\rho)
    =
    \hat\sigma^2(I-\rho W)^{-1}(I-\rho W')^{-1},
    \qquad
    \rho \in \mathcal{R}.
\]

In summary, peer effects do not create a need for different estimators if one is interested in direct effects, although accounting for social multiplier effects requires additional assumptions to identify $\rho$. Inference is feasible with correlation-robust standard errors.



\section*{Appendix C: Simulation Results}

The variables are generated as follows: a random vector $V$, whose components are uniform, triangular, and sum of three uniforms, is used to construct
\begin{equation*}
    \begin{pmatrix} X_1 \\ X_2 \\ X_3 \end{pmatrix} = \begin{pmatrix} 0.7 && 0.3 && 0 \\ -0.1 && 1 && 0.4 \\ 0 && -0.6 && 0.7 \end{pmatrix} \begin{pmatrix} V_1 \\ V_2 \\ V_3 \end{pmatrix}
\end{equation*}
so that there is a non-trivial correlation structure among the components of $X$. In the baseline, $\mbox{corr}(X_1, X_2)=0.28$, $\mbox{corr}(X_1, X_3)=-0.26$, and $\mbox{corr}(X_2, X_3)=-0.32$, though the results exhibit similar patterns for weaker or stronger correlations. The variances are normalized to one. The first two variables are observed, but the last one is not, \textit{i.e.}, $X^o = (X_1, X_2)$ and $X^u = X_3$. \medskip

The propensity score follows a logistic distribution with argument $X \beta$, where $\beta = \begin{pmatrix} 1 & 1 & \beta_3 \end{pmatrix}'$, and the treatment status is then drawn conditional on $X$. The parameter $\beta_3$ controls selection on unobservables and takes value in $\{0, 0.5, 1\}$. In the first case the probability of being treated does not change with $X_3$; in the last case the unobserved variable is on a similar footing as each of the observed ones. The performance of traditional methods that cannot account for the unobserved component is expected to deteriorate as $\beta_3$ increases. \medskip

The outcome equation is given by $y = 5 + \mbox{CATE}(x) T + g(x) + \varepsilon$, $\varepsilon \sim \mathcal{N}(0, 1)$. I explore three specifications:  
\begin{itemize}
    \item Homogeneous treatment effects ($\mbox{CATE}(x) \equiv 1$) with linear impact of unobservables ($g$ is linear in $x$)
    \item Heterogeneous treatment effects ($\mbox{CATE}(x) = 2 \Phi(-x_1+x_2+x_3)$) with linear impact of unobservables ($g$ is linear in $x$)
    \item Heterogeneous treatment effects ($\mbox{CATE}(x) = 2 \Phi(-x_1+x_2+x_3)$) with quadratic impact of unobservables ($g$ contains both a linear term in $x$ and a quadratic term ${x_3}^2+x_2 x_3$)
\end{itemize}


The network formation process uses $w(x)=e^{-\frac{1}{2} x^2}$.\footnote{Other specifications such as $w(x)=\mathds{1}_{x<1}$ or $w(x)=\max\{1-x, 0\}$ deliver similar results.} The baseline sample size is $n=500$ and $s_n$ is calibrated so that the average number of friends is roughly five to six, which is around the number of close friends people often report, as in the application. As $n$ rises, ${s_n}^d$ evolves at the rate $n / \ln(n)$. \medskip 

The researcher controls for observables using kernel weights that multiply network weights (Remark \ref{kernel_product_remark}). People form friendship links based on $X_2$ and $X_3$. As a result, only $X_1$ is unaccounted for in network weights, but a researcher observing $X_2$ may want to further include it in the kernel weights. I consider both possibilities (referred to as base-control and over-control cases below). Alternative methods (\textit{e.g.}, OLS) always make use of all observables. \medskip

The results for all estimators and sample sizes $n=500, 2000$ are reported in Appendix C. For a simple and representative summary, the results for the specifications $M=1, c=2$ (orange and yellow lines, respectively) and OLS (blue line) are graphically depicted below for the 3 types of outcome equations and $n=500$. \medskip

\begin{figure}[ht!]
    {\includegraphics[width=2.3in]{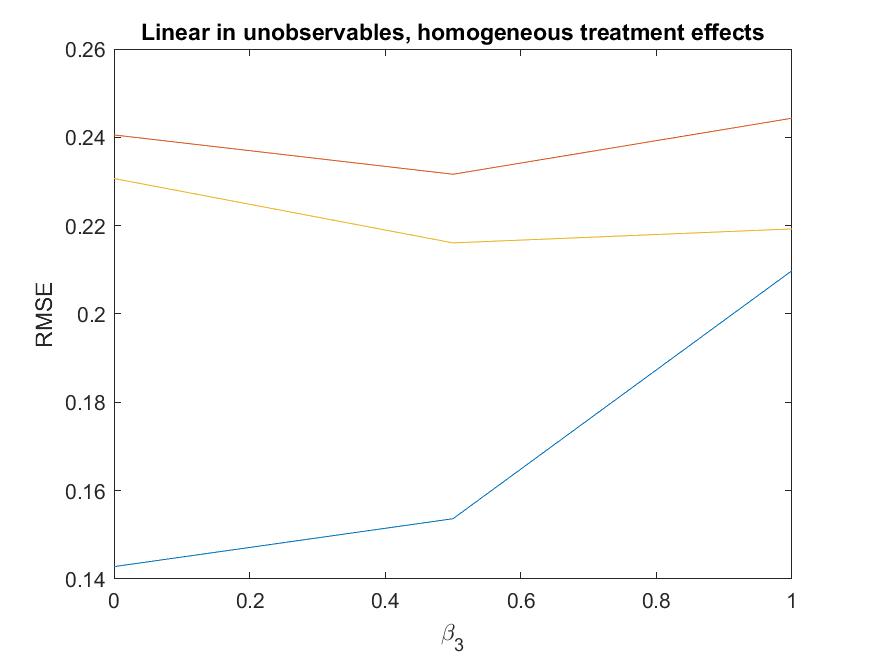}}\includegraphics[width=2.3in]{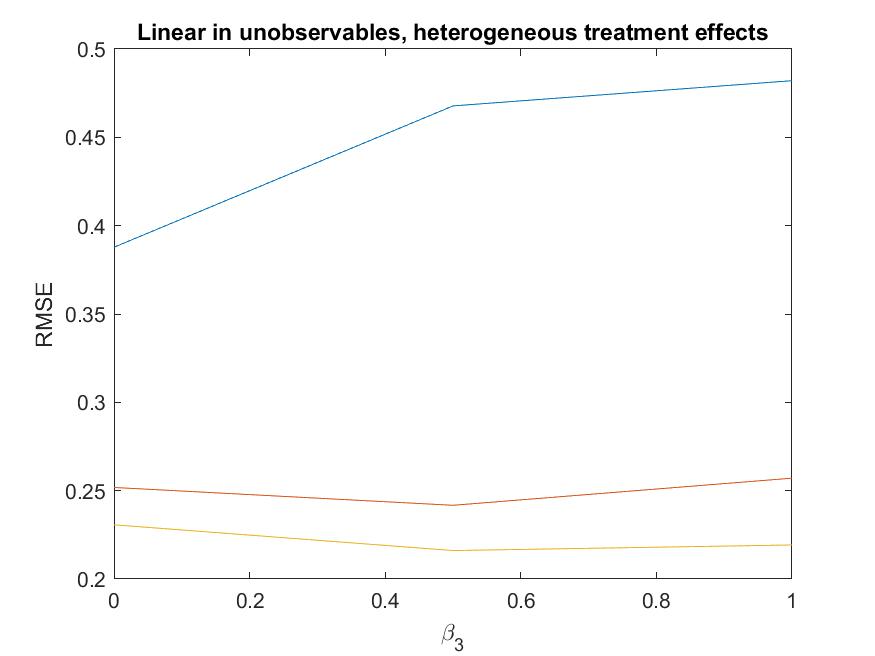}\includegraphics[width=2.3in]{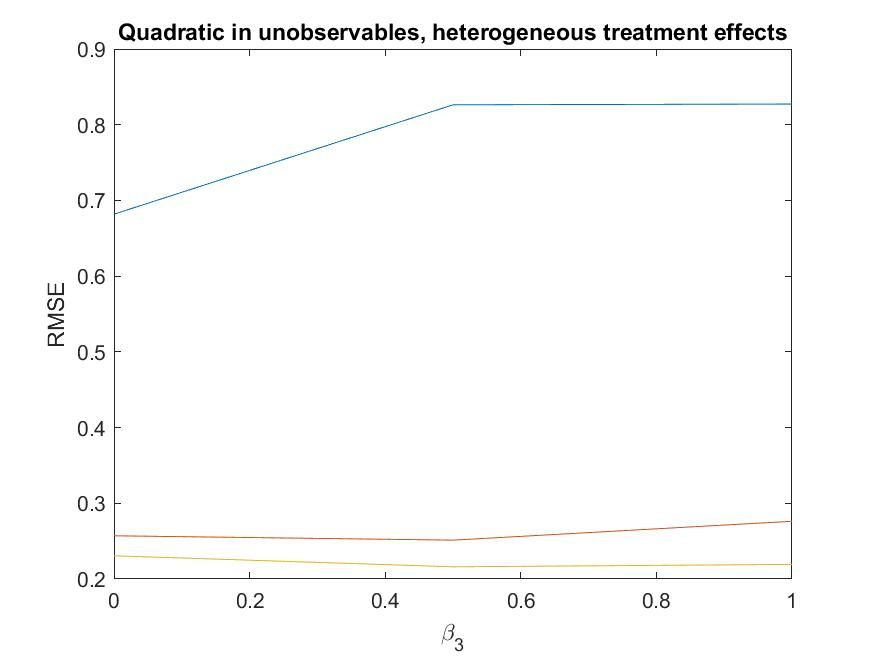}
    {\includegraphics[width=2.3in]{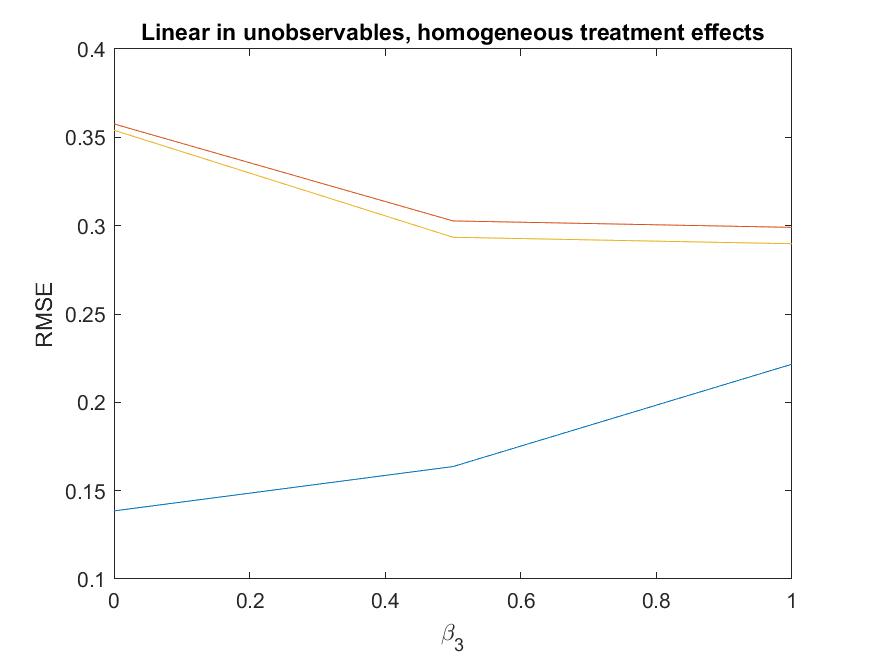}}\includegraphics[width=2.3in]{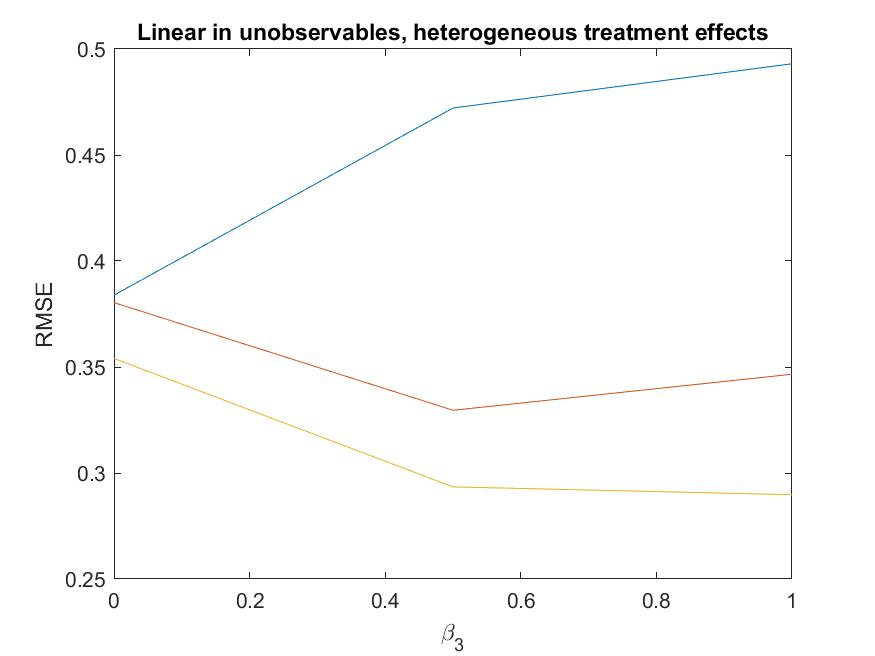}\includegraphics[width=2.3in]{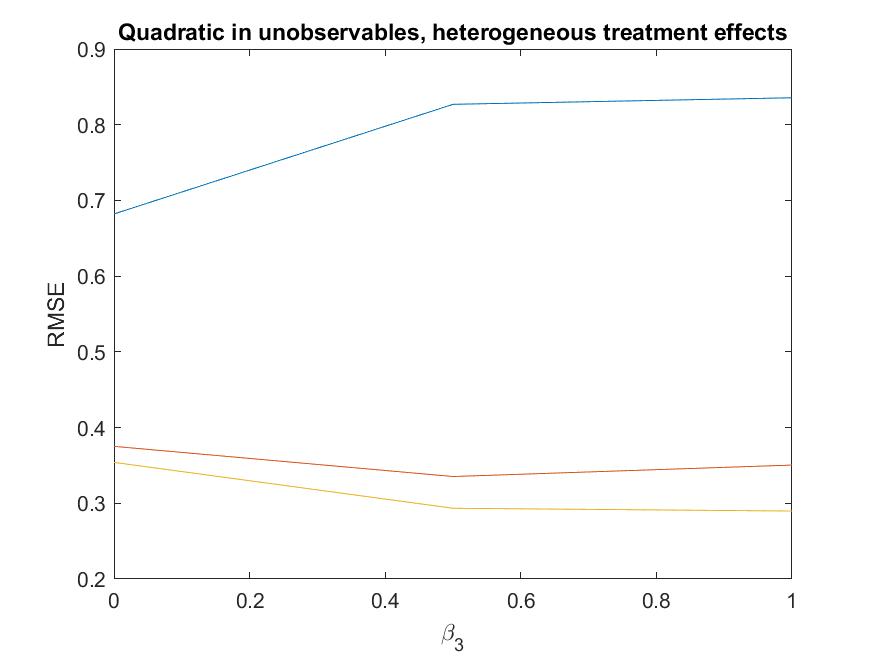}
\caption{RMSE of OLS, $\ATE$ estimators using $M=1, c=2$. The first row refers to the base-control specification; the second row refers to the over-control specification.}
\end{figure}
In the linear homogeneous case the OLS estimator performs well, especially when there is no selection on unobservables. Its performance deteriorates as $\beta_3$ increases and it reaches an RMSE similar to that of the proposed estimators when the selection on the unobserved variable is on par with the level of selection on observables. Furthermore, it behaves poorly when the unobservables become more prevalent in the functional form or the treatment response. \medskip

In contrast, estimators that leverage network information are able to perform similarly regardless of the strength of selection in unobservables and dominate OLS across specifications unless we force homogeneity and linearity. In the presence of unobserved confounding the properties of OLS become quickly unappealing, especially when the unobserved variable has nonlinear impacts on the outcome or the treatment effect. Other methods based on estimating the propensity score also fail to deliver reliable estimates of the ATE as they can only accommodate heterogeneity and nonlinearities that arise from observed confounders. They are even dominated by OLS and \textit{a fortiori} by the proposed estimators. \medskip

Finally the simulations provide some guidance for empirical decisions. The choice of $M$ or $c$ does not appear to have a strong influence on the performance of the estimator. Over-controlling leads to a general increase in RMSE, but the estimators still perform well and improve substantially over OLS in nonlinear or heterogeneous settings. Therefore, it is advisable to separately control for observables whose relevance to network formation is uncertain.

\begin{table}[!ht]
\centering
\caption{RMSE of ATE estimators (n=500)}
\begin{tabular}{|l|l|l|l|l|l|l|l|l|l|l|l|}\hline
y & $\beta_3$ & M=1  & M=2  & M=3  & M=4 & c=2  & c=3 & c=4 & OLS  & Strat & IPW \\ \hline
\multirow{3}{*}{A} & 0 & 0.24 & 0.25 & 0.24 & 0.24 & 0.23 & 0.23 & 0.38 & 0.14 & 0.21 & 0.43 \\
& 0.5 & 0.23 & 0.23 & 0.22 & 0.24 & 0.22 & 0.26 & 0.43 & 0.15 & 0.23 & 0.44 \\
& 1 & 0.24 & 0.23 & 0.23 & 0.24 & 0.22 & 0.26 & 0.43 & 0.21 & 0.29 & 0.36 \\
\multirow{3}{*}{B} & 0 & 0.25 & 0.26 & 0.25 & 0.26 & 0.24 & 0.22 & 0.38 & 0.39 & 0.47 & 0.48 \\
& 0.5 & 0.24 & 0.25 & 0.25 & 0.28 & 0.23 & 0.23 & 0.40 & 0.47 & 0.54 & 0.55 \\
& 1 & 0.26 & 0.26 & 0.26 & 0.28 & 0.24 & 0.25 & 0.41 & 0.48 & 0.57 & 0.52 \\
\multirow{3}{*}{C} & 0 & 0.26 & 0.29 & 0.28 & 0.30 & 0.25 & 0.21 & 0.37 & 0.68 & 0.76 & 0.70 \\
& 0.5 & 0.25 & 0.28 & 0.29 & 0.33 & 0.25 & 0.23 & 0.39 & 0.83 & 0.90 & 0.85 \\
& 1 & 0.28 & 0.30 & 0.30 & 0.34 & 0.26 & 0.24 & 0.41 & 0.83 & 0.91 & 0.84 \\ \hline
\end{tabular}
\end{table}

\begin{table}[!ht]
\centering
\caption{RMSE of ATE estimators (n=2000)}
\begin{tabular}{|l|l|l|l|l|l|l|l|l|l|l|l|}\hline
y & $\beta_3$ & M=1  & M=2  & M=3 & M=4 & c=2 & c=3 & c=4 & OLS  & Strat & IPW  \\ \hline
\multirow{3}{*}{A} & 0 & 0.12 & 0.14 & 0.13 & 0.12 & 0.11 & 0.16 & 0.32 & 0.07 & 0.12 & 0.20 \\
& 0.5 & 0.13 & 0.13 & 0.12 & 0.12 & 0.11 & 0.19 & 0.36 & 0.08 & 0.13 & 0.22 \\
& 1 & 0.14 & 0.13 & 0.12 & 0.12 & 0.11 & 0.20 & 0.36 & 0.11 & 0.15 & 0.18 \\
\multirow{3}{*}{B} & 0 & 0.12 & 0.18 & 0.15 & 0.15 & 0.14 & 0.14 & 0.30 & 0.37 & 0.44 & 0.39 \\
& 0.5 & 0.12 & 0.16 & 0.15 & 0.15 & 0.12 & 0.16 & 0.33 & 0.45 & 0.51 & 0.47 \\
& 1 & 0.13 & 0.17 & 0.16 & 0.16 & 0.13 & 0.17 & 0.34 & 0.46 & 0.52 & 0.47 \\
\multirow{3}{*}{C} & 0 & 0.13 & 0.19 & 0.17 & 0.17 & 0.15 & 0.13 & 0.30 & 0.67 & 0.75 & 0.68 \\
& 0.5 & 0.12 & 0.20 & 0.18 & 0.19 & 0.15 & 0.15 & 0.33 & 0.81 & 0.88 & 0.83 \\
& 1 & 0.13 & 0.20 & 0.19 & 0.19 & 0.15 & 0.16 & 0.33 & 0.82 & 0.89 & 0.83 \\ \hline
\end{tabular}
\caption*{RMSE of the estimator using friends up to order $M$ (columns 1-4), of the estimator using people with at least $c$ friends in common (columns 5-7), of OLS, from stratification based on propensity scores, and of the inverse-propensity weighted estimator. $y$ refers to the type of outcome equation (A: homogeneous effects with linear specification, B: heterogeneous effects, C: heterogeneous effects and quadratic specification).}
\end{table}

\pagebreak

\begin{table}[!ht]
\centering
\caption{RMSE of ATE estimators (over-controlling; n=500)}
\begin{tabular}{|l|l|l|l|l|l|l|l|l|}\hline
y & $\beta_3$ & M=1  & M=2  & M=3 & M=4 & c=2 & c=3 & c=4 \\ \hline
\multirow{3}{*}{A} & 0 & 0.36 & 0.38 & 0.37 & 0.38 & 0.35 & 0.19 & 0.28 \\
& 0.5 & 0.30 & 0.32 & 0.32 & 0.32 & 0.29 & 0.19 & 0.32 \\
& 1 & 0.30 & 0.31 & 0.32 & 0.33 & 0.29 & 0.20 & 0.34 \\
\multirow{3}{*}{B} & 0 & 0.38 & 0.38 & 0.38 & 0.39 & 0.36 & 0.19 & 0.28\\
& 0.5 & 0.33 & 0.34 & 0.34 & 0.36 & 0.31 & 0.18 & 0.32\\
& 1 & 0.35 & 0.35 & 0.34 & 0.37 & 0.33 & 0.20 & 0.31\\
\multirow{3}{*}{C} & 0 & 0.38 & 0.38 & 0.37 & 0.40 & 0.36 & 0.19 & 0.28\\
& 0.5 & 0.34 & 0.34 & 0.35 & 0.39 & 0.32 & 0.18 & 0.32\\
& 1 & 0.35 & 0.36 & 0.36 & 0.39 & 0.34 & 0.21 & 0.34\\ \hline
\end{tabular}
\end{table}
\begin{table}[!ht]
\centering
\caption{RMSE of ATE estimators (over-controlling; n=2000)}
\begin{tabular}{|l|l|l|l|l|l|l|l|l|}\hline
y & $\beta_3$ & M=1  & M=2  & M=3  & M=4 & c=2  & c=3 & c=4\\ \hline
\multirow{3}{*}{A} & 0 & 0.22 & 0.24 & 0.21 & 0.21 & 0.21 & 0.10 & 0.20\\
& 0.5 & 0.17 & 0.19 & 0.18 & 0.17 & 0.17 & 0.11 & 0.25\\
& 1 & 0.18 & 0.20 & 0.18 & 0.17 & 0.18 & 0.11 & 0.24\\
\multirow{3}{*}{B} & 0 & 0.24 & 0.24 & 0.21 & 0.21 & 0.23 & 0.10 & 0.20\\
& 0.5 & 0.20 & 0.21 & 0.19 & 0.19 & 0.19 & 0.10 & 0.24\\
& 1 & 0.20 & 0.21 & 0.19 & 0.19 & 0.19 & 0.10 & 0.24\\
\multirow{3}{*}{C} & 0 & 0.23 & 0.25 & 0.22 & 0.21 & 0.23 & 0.09 & 0.20\\
& 0.5 & 0.22 & 0.23 & 0.21 & 0.21 & 0.21 & 0.10 & 0.24\\
& 1 & 0.22 & 0.24 & 0.21 & 0.21 & 0.21 & 0.11 & 0.24 \\ \hline
\end{tabular}
\caption*{RMSE of the estimator using friends up to order $M$ (columns 1-4), of the estimator using people with at least $c$ friends in common (columns 5-7) in the over-controlling case. $y$ refers to the type of outcome equation (A: homogeneous effects with linear specification, B: heterogeneous effects, C: heterogeneous effects and quadratic specification).}
\end{table}

\end{document}